\documentclass[11pt]{amsart}

\usepackage{xcolor}
\usepackage[colorlinks,linkcolor=red,citecolor=purple]{hyperref}

\usepackage{fullpage}

\usepackage{mathpazo}

\usepackage[OT1]{fontenc}
\usepackage{ragged2e}

\usepackage{amsmath,amssymb,amsthm,mathtools}
\usepackage{todonotes}
\usepackage{xifthen}
\usepackage{cleveref}
\usepackage{enumitem}
\usepackage{algorithm}
\usepackage{algpseudocode}

\newcommand{\tp}[1]{\left(#1\right)}
\newcommand{\sqtp}[1]{\left[#1\right]}
\newcommand{\floor}[1]{{\lfloor #1 \rfloor}}
\newcommand{\ceil}[1]{{\lceil #1 \rceil}}

\newcommand{\ul}{\underline}
\newcommand{\ol}{\overline}
\newcommand{\eps}{\varepsilon}
\newcommand{\set}[1]{\left\{#1\right\}}
\newcommand{\abs}[1]{{\left| #1 \right|}}

\theoremstyle{plain}
\newtheorem{theorem}{Theorem}
\newtheorem{lemma}[theorem]{Lemma}
\theoremstyle{definition}

\def\*#1{\mathbf{#1}}
\def\-#1{\mathrm{#1}}
\def\+#1{\mathcal{#1}}
\def\=#1{\mathbb{#1}}

\newcommand{\setG}{\+G^{\-{bip}}_{n,\Delta}}
\newcommand{\clusterOfColoring}[1][]{ \ifthenelse{\isempty{#1}}
  {\+C_X(G)}
  {\+C_{#1}(G)} }
\newcommand{\clusterOfIS}[1][]{ \ifthenelse{\isempty{#1}}
  {\+I_{\+X}(G)}
  {\+I_{#1}(G)}
}
\newcommand{\cmid}{:\,}
\newcommand{\lowerBoundOfDeltaForIS}{53}
\newcommand{\conf}[1][]{ \ifthenelse{\isempty{#1}}
  {\omega_{\ol \gamma}}
  {\omega_{\ol #1}}
}
\newcommand{\aabs}[1]{{\left| \ol{#1} \right|}}
\newcommand{\tOfIS}{2.9}
\newcommand{\valueOfZeta}{1.28}

\newcommand{\sOfColorings}{s}

\newcommand{\graphClassAtHighFugacity}{{\+G^{\Delta}_{\alpha,\beta}}}
\newcommand{\graphClassAtLowFugacity}{{\+G^{\Delta}_{\alpha,\alpha,\beta}}}
\newcommand{\graphClassOfColorings}{{\+G^{\Delta}_{q,\sOfColorings,\alpha,\beta}}}
\newcommand{\valueOfS}{\frac{1}{18\ol q^5}}
\newcommand{\valueOfAlphaForColorings}{\frac{1}{\Delta^{1/2}}}
\newcommand{\upperBoundOfAlphaForColorings}{\frac{1}{10\ol q^5}}
\newcommand{\valueOfBetaForColorings}{\frac{\Delta^{1/2}}{3}}

\newcommand{\lowerBoundOfDeltaForColorings}{100\ol q^{10}}

\newcommand{\uniquenessThreshold}{\lambda_c(\Delta)}
\newcommand{\valueOfUniquenessThreshold}{\frac{(\Delta-1)^{\Delta-1}}{(\Delta-2)^\Delta}}

\newcommand{\lowerBoundOfLambdaInOmegaForm}{\widetilde\Omega\tp{\frac{1}{\Delta}}}
\newcommand{\valueOfLambdaL}{\frac{\tp{\ln \Delta}^4}{\Delta}}
\newcommand{\valueOfAlphaAtLowFugacity}{\frac{(\ln \Delta)^2}{\Delta}}
\newcommand{\valueOfBetaAtLowFugacity}{\frac{1}{3\alpha}}

\begin{document}

\title{Counting independent sets and colorings on random regular bipartite graphs}

\author{Chao Liao}
\address[Chao Liao]{CSE, Shanghai
Jiao Tong University, No.800 Dongchuan Road, Minhang District, Shanghai, China.}
\email{chao.liao.95@gmail.com}

\author{Jiabao Lin}
\address[Jiabao Lin]{ITCS, Shanghai University of Finance and Economics, No.100
Wudong Road, Yangpu District, Shanghai, China.}
\email{lin.jiabao@mail.shufe.edu.cn}

\author{Pinyan Lu}
\address[Pinyan Lu]{ITCS, Shanghai University of Finance and Economics, No.100
Wudong Road, Yangpu District, Shanghai, China.}
\email{lu.pinyan@mail.shufe.edu.cn}

\author{Zhenyu Mao}
\address[Zhenyu Mao]{ITCS, Shanghai University of Finance and Economics, No.100
Wudong Road, Yangpu District, Shanghai, China.}
\email{zhenyu.mao.17@gmail.com}

\begin{abstract}
We give a fully polynomial-time approximation scheme (FPTAS) to count the number of independent sets on almost every $\Delta$-regular bipartite graph if $\Delta\ge \lowerBoundOfDeltaForIS$.
In the weighted case, for all sufficiently large integers $\Delta$ and weight parameters $\lambda=\lowerBoundOfLambdaInOmegaForm$, we also obtain an FPTAS on almost every $\Delta$-regular bipartite graph.
Our technique is based on the recent work of Jenssen, Keevash and Perkins (SODA, 2019) and we also apply it to confirm an open question raised there:
For all $q\ge 3$ and sufficiently large integers $\Delta=\Delta(q)$, there is an FPTAS to count the number of $q$-colorings on almost every $\Delta$-regular bipartite graph.
\end{abstract}
\maketitle

\section{Introduction}

Counting independent sets on bipartite graphs (\#BIS) plays a significant role in the field of approximate counting.
A wide range of counting problems in the study of counting CSPs \cite{DGJ10,BDGJM13,GGY17} and spin systems \cite{GJ12,GJ15,GSVY16,CGGGJSV16}, have been proved to be \#BIS-equivalent or \#BIS-hard under approximation-preserving reductions (AP-reductions) \cite{DGGJ04}.
Despite its great importance, it is still unknown whether \#BIS admits a fully polynomial-time approximation scheme (FPTAS) or it is as hard as counting the number of satisfying assignments of Boolean formulas (\#SAT) under AP-reduction.


In this paper, we consider the problem of approximating \#BIS (and its weighted version) on random regular biparite graphs.
Random regular bipartite graphs frequently appear in the analysis of hardness of counting independent sets \cite{MWW09,DFJ02,Sly10,SS12,GSVY16}.
Therefore, understanding the complexity of \#BIS on such graphs is potentially useful for gaining insights into the general case.
Let $Z(G,\lambda) = \sum_{I\in \+I(G)}\lambda^{\abs I}$ where $\+I(G)$ is the set of all independent sets of a graph $G$ and $\lambda > 0$ is the weight parameter.
This function also arises in the study of the hardcore model of lattice gas systems in statistical mechanics.
Hence we usually call $Z(G,\lambda)$ the partition function of the hardcore model with fugacity $\lambda$.

In the case where input graphs are allowed to be nonbipartite, the approximability for counting the number of independent sets (\#IS) is well understood.
Exploiting the correlation decay properties of $Z(G,\lambda)$, Weitz \cite{Weitz06} presented an FPTAS for graphs of maximum degree $\Delta$ at fugacity $\lambda < \uniquenessThreshold=\valueOfUniquenessThreshold$.
On the hardness side, Sly \cite{Sly10} proved that, unless $\text{NP}=\text{RP}$, there is a constant $\eps=\eps(\Delta)$ that no polynomial-time approximation scheme exists for $Z(G,\lambda)$ on graphs of maximum degree $\Delta$ at fugacity $\lambda_c(\Delta) < \lambda < \lambda_c(\Delta) + \eps(\Delta)$.
Later, this result was improved at any fugacity $\lambda > \uniquenessThreshold$ \cite{SS12,GSV16}.
In particular, these results state that if $\Delta\le 5$, there is an FPTAS for \#IS on graphs of maximum degree $\Delta$, otherwise there is no efficient approximation algorithm unless $\text{NP}=\text{RP}$.

The situation is different on bipartite graphs.
To the best of our knowledge, no NP-hardness result is known even on graphs with unbounded degree.
Surprisingly, Liu and Lu \cite{LL15} designed an FPTAS for \#BIS which only requires one side of the vertex partition to be of
maximum degree $\Delta\leq 5$.
On the other hand, it is \#BIS-hard to approximate $Z(G,\lambda)$ at fugacity $\lambda>\lambda_c(\Delta)$ on biparite graphs of maximum degree $\Delta\ge 3 $ \cite{CGGGJSV16}.

Recently, Helmuth, Perkins, and Regts \cite{HPR19} developed a new approach via the polymer model and gave efficient counting and sampling algorithms for the hardcore model at high fugacity on certain finite regions of the lattice $\=Z^d$ and on the torus $(\=Z/ n\=Z)^d$.
Their approach is based on a long line of work \cite{PS75, PS76, KP86, Bar16, BS16, PR17}.
Shortly after that, Jessen, Keevash, and Perkins \cite{JKP19} designed an FPTAS for the hardcore model at high fugacity on bipartite expander graphs of bounded degree. And they further extended the result to random $\Delta$-regular bipartite graphs with $\Delta\ge 3$ at fugacity $\lambda>(2e)^{250}$. This is the first efficient algorithm for the hardcore model at fugacity $\lambda>\lambda_c(\Delta)$ on random regular bipartite graphs. A natural question is, can we  design FPTAS for lower fugacity and in particular the problem \#BIS on random regular bipartite graphs?
Indeed, we obtain such results.
Let $\setG$ denote the set of all $\Delta$-regular bipartite graphs with $n$ vertices on both sides.

\begin{theorem}\label{thm:BIS}
   For $\Delta\ge 53$ and fugacity $\lambda\ge 1$, with high probability (tending to $1$ as $n\to \infty$) for a graph $G$ chosen uniformly at random from $\setG$, there is an FPTAS for the partition function $Z(G,\lambda)$.
\end{theorem}

\begin{theorem}\label{thm:weightedBIS}
  For all sufficiently large integers $\Delta$ and fugacity $\lambda=\lowerBoundOfLambdaInOmegaForm$, with high probability (tending to $1$ as $n\to \infty$) for a graph $G$ chosen uniformly at random from $\setG$, there is an FPTAS for the partition function $Z(G,\lambda)$.
\end{theorem}

For notational convenience, we use the term ``on almost every $\Delta$-regular bipartite graph'' to denote that a property holds with high probability (tending to $1$ as $n\to\infty$) for randomly chosen graphs from $\setG$.

Counting proper $q$-colorings on a graph is another extensively
studied problem in the field of approximate counting~\cite{jerrum1995very,bubley1997path,bubley1999approximately,dyer2003randomly,hayes2003non,hayes2003randomly,molloy2004glauber,dyer2006randomly,hayes2006coupling,gamarnik2012correlation,dyer2013randomly,lu2013improved,GLLZ18}, which is also shown to be \#BIS-hard but unknown to be \#BIS-equivalent \cite{DGGJ04}.
In general graphs, if the number of colors $q$ is no more than the maximum degree $\Delta$, there may not be any proper coloring over the graph. Therefore, approximate counting is studied in the range that $q\geq \Delta+1$. It was conjectured that there is an FPTAS if $q\geq \Delta+1$, but the current best result is $q\geq \alpha \Delta+1$ with a constant $\alpha$ slightly below $\frac{11}{6}$~\cite{vigoda2000improved,ChenDMPP19}. The conjecture was only confirmed for the special case $\Delta=3$ \cite{LuYZZ17}.

On bipartite graphs, the situation is quite different.
For any $q\geq 2$, we know that there always exist proper $q$-colorings for every bipartite graph.
So it is natural to wonder under which relations between $q$ and $\Delta$ there is an FPTAS to count the number of $q$-colorings on biparite graphs.
Using a technique analogous to that for \#BIS,
we obtain an FPTAS to count the number of $q$-colorings on random $\Delta$-regular bipartite graphs for all sufficiently large integers $\Delta=\Delta(q)$ for any $q\ge 3$.
\begin{theorem}\label{thm:q-coloring}
    For $q\ge 3$ and $\Delta \ge \lowerBoundOfDeltaForColorings$ where $\ol q = \ceil{q/2}$, with high probability (tending to $1$ as $n\to \infty $) for a graph chosen uniformly at random from $\setG$, there is an FPTAS to count the number of $q$-colorings.
\end{theorem}

This result confirms a conjecture in \cite{JKP19}.

\subsection{Our Technique}
The classical approach to designing approximate counting algorithms is random sampling
via Markov chain Monte Carlo (MCMC). However, it is known that the Markov chains are slowly mixing on random
bipartite graphs for both independent set and coloring if the degree $\Delta$ is not too small.
Taking \#BIS as an example, a typical independent set of a random regular bipartite graph of degree at least $6$ is
unbalanced: it either chooses most of its vertices from the left side or the right side. Thus, starting from an independent set with most vertices from the left side, a Markov chain is unlikely to reach an independent set with most of its vertices from the right side in polynomial time.

Even so, a recent beautiful work exactly makes use of the above separating property to design approximately counting algorithm~\cite{JKP19}. By making the fugacity $\lambda > (2e)^{250}$ sufficiently large, they proved that most contribution of the partition function comes from extremely unbalanced independent sets, those which occupy almost no vertices on one side and almost all vertices on the other side. In particular, for a bipartite graph $G=(\+L,\+R,E)$ with $n$ vertices on both sides, they identified two independent sets $I=\+L$ and $I=\+R$ as ground states as they have the largest weight $\lambda^{n}$ among all the independent sets.
They proved that one only needs to sum up the weights of states which are close to one of the ground states, for no state is close to both ground states and the contribution from the states which are far away from both ground states is exponentially small.

However, the ground state idea cannot be directly applied to counting independent sets and counting colorings since each valid configuration is of the same weight. We extend the idea of ground states to ground clusters, which is not a single configuration but a family of configurations.
For example, we identify two ground clusters for independent sets, those which are entirely chosen from vertices on the left side and those which are entirely chosen entirely from vertices on the right side.
If a set of vertices is entirely chosen from vertices on one side, it is obviously an independent set.
Thus each cluster contains $2^n$ different independent sets.
Similarly, we want to prove that we can count the configurations which are close to one of the ground clusters and then add them up.
For counting colorings, there are multiple ground clusters indexed by a subset of colors $\emptyset\subsetneq X\subsetneq [q]$: colorings which color $\+L$ only with colors from $X$ and color $\+R$ only with colors from $[q]\setminus X$.

Unlike the ground states in \cite{JKP19}, our ground clusters may overlap with each other and some configurations are close to more than one ground clusters.
In addition to proving that the number of configurations which are far away from all ground clusters are exponentially small, we also need to prove that the number of double counted configurations are small.

After identifying ground states and with respect to a fixed ground state,  Jessen, Keevash, and
Perkins \cite{JKP19} defined a polymer model representing deviations from the ground state and rewrote the original partition function as a polymer partition function. We follow this idea and define a polymer model representing deviations from a ground cluster. However, deviation from a ground cluster is much subtler than deviation from a single ground state. For example, if we define polymer as connected components from the deviated vertices in the graph, we cannot recover the original partition function from the polymer partition function.
We overcome this by defining polymer as connected components in graph $G^2$, where an edge of $G^2$ corresponds to a path of length at most $2$ in the original graph.
Here, a compatible set of polymers also corresponds to a family of configurations in the original problem, while it corresponds to a single configuration in \cite{JKP19}.

It is much more common in counting problems that most contribution is from a neighborhood of some clusters rather than a few isolated states. So, we believe that our development of the technique makes it suitable for a much broader family of problems.

\subsection*{Independent work}
Towards the end of this project, we learned that the authors of ~\cite{JKP19} obtained similar results in their upcoming journal version submission.

\section{Preliminaries}

In this section, we review some basic definitions and concepts, introduce necessary notations and set up some facts and tools.

\subsection{Independent sets and colorings}
All graphs considered in this paper are unweighted, undirected, with no loops but may have multiple edges\footnote{There is no essential difference from keeping the graphs simple. We allow multiple edges just for writing convenience.}.
Let $G=(V,E)$ be a graph.
We use $d_G(u,w)$ to denote the distance between two vertices $u,w$ in the graph $G$.
For $\emptyset\subsetneq U,W\subseteq V$, define $d_G(U,W) = \min_{u\in U,w\in W} d_G(u,w)$.
Let $U\subseteq V$ be a nonempty set.
We define $N_G(U) = \set{v\in V\cmid d_G(\set{v},U)=1}$ to be the neighborhood of $U$ and emphasize that $N_G(U)\cap U=\emptyset$.
We use $G[U]$ to denote the induced subgraph of $G$ on $U$.
Let $E^2$ be the set of unordered pairs $(u,v)$ such that $u\neq v$ and $d_G(u,v)\le 2$.
We define $G^2$ to be the graph $(V,E^2)$. It is clear that if the maximum degree of $G$ is at most $\Delta$, then the maximum degree of $G^2$ is at most $\Delta^2$.

An independent set of the graph $G$ is a subset $U\subseteq V$ such that $(u,w)\not\in E$ for any $u,w\in U$.
We use $\+I(G)$ to denote the set of all independent sets of $G$.
The weight of an independent set $I$ is $\lambda^{\abs I}$ where $\lambda>0$ is a paramter called fugacity.
We use $Z(G, \lambda) = \sum_{I\in \+I(G)}\lambda^{\abs I}$ to denote the partition function of the graph $G$. Clearly, $Z(G, 1)$ is the number of indepndent sets of $G$.

For any positive integer $i$, we use $[i]$ to denote the set $\set{1,2,\ldots,i}$.
Let $q \ge 3$ be an integer.
Define $\ul q = \floor{q/2}$ and $\ol q = \ceil{q/2}$.
A coloring $\sigma:V\to [q]$ over the graph $G$ is a mapping which assigns to each vertex of $G$ a color from $[q]$.
We say $\sigma$ is proper if $\sigma(u)\neq \sigma(v)$ for any edge $(u,v)\in E$.
We use $\+C(G)$ to denote the set of all proper colorings over $G$.
Sometimes we need to consider the rescriction of a coloring and we use $\sigma|_U$ to denote the coloring obtained by restricting $\sigma$ over a subset $U\subseteq V$.
Whenever $G=(\+L,\+R,E)$ is a bipartite graph and $\sigma$ is coloring over $G$, we simply write $\sigma_\+X$ instead of $\sigma|_\+X$ for all $\+X\in\set{\+L,\+R}$.
For a number of disjoint sets $S_1,S_2,\ldots,S_k$, we use $\sqcup_{i=1}^k S_i$ to denote their union and stress the fact that they are disjoint.
For a number of colorings $\sigma_1:V_1\to[q],\sigma_2:V_2\to[q],\ldots,\sigma_k:V_k\to [q]$, if $V_i$ and $V_j$ are disjoint for any $1\le i \neq j\le k$, then $\cup_{i=1}^k\sigma_i$ is the coloring over $\sqcup_{i=1}^k V_i$ such that its resctriction over $V_i$ is $\sigma_i$ for any $1\le i \le k$.

For two positive real numbers $a$ and $b$, we say $a$ is an $\eps$-relative approximation to $b$ for some $\eps > 0$ if $\exp(-\eps)b \le a \le \exp(\eps)b$, or equivalently $\exp(-\eps)a\le b \le \exp(\eps)a$.
A fully polynomial-time approximation scheme (FPTAS) is an algorithm that for every $\eps > 0$ outputs an $\eps$-relative approximation to $Z(G)$ in time $\tp{\abs{G}/\eps}^C$ for some constant $C>0$, where $Z(G)$ is some quantity, like the number of independent sets, of graphs $G$ that we would like to compute.

\subsection{Random regular bipartite graphs}

We follow the model of random regular bipartite graphs in \cite{MWW09}.
Let $\Delta$ be a positive integer.
We use $G\sim \setG$ to denote sampling a bipartite graph $G$ in the following way.
At the beginning, the two sides of $G$ both have exactly $n$ vertices and there are no edges between them.
In the $i$-th round, we sample a perfect matching $M_i$ over the complete bipartite graph $K_{n,n}$ uniformly at random and independently of previous rounds.
We repeat this process for $\Delta$ rounds and add the edges in $M_1,M_2,\ldots,M_\Delta$ to the graph $G$.
We do not merge multiple edges in $G$ to keep it $\Delta$-regular.
We remark that this distribution of random graphs is contiguous with a uniformly random $\Delta$-regular simple (without multiple edges) bipartite graph, which implies that \Cref{lem:conditionOfExpanders} and similar results also apply to the latter distribution.
See \cite{1factor} for more information.
In the following, we discuss the property of random regular bipartite graphs.

We say a $\Delta$-regular bipartite graph $G=(\+L,\+R,E)$ with $n$ vertices on both sides is an $(\alpha,\beta)$-expander if for all subsets $U\subseteq \+L$ or $U\subseteq \+R$ with $\abs{U}\le \alpha n$, $\abs{N(U)}\ge \beta \abs{U}$.
This property is called the expansion property of $G$.
We use $\graphClassAtHighFugacity$ to denote the set of all $\Delta$-regular bipartite $(\alpha,\beta)$-expander.
The following lemma states that under certain conditions almost every $\Delta$-regular graph is an $(\alpha,\beta)$-expander.

\begin{lemma}[\cite{Bassalygo1981}] \label{lem:conditionOfExpanders}
  If $0 < \alpha < 1/\beta < 1$ and
  $\displaystyle\Delta > \frac{H(\alpha) + H(\alpha \beta)}{H(\alpha) - \alpha \beta H(1/\beta)}$,
  then
  \[\lim_{n\to\infty}\Pr_{G\sim \setG}\sqtp{G\in \graphClassAtHighFugacity} = 1.\]
\end{lemma}

In addition to the expansion property, random regular graphs may also have the following property.
For $0<a,b<1$, we say a bipartite graph $G=(\+L,\+R,E)$ with $n$ vertices on both sides has the $(a,b)$-cover property if $\abs{N_G(U)}> (1-b)n$ for all $U\subseteq \+L$ or $U\subseteq \+R$ with $\abs{U}\ge an$.

\subsection{The polymer model} \label{sec:polymer-model}

We follow the way in \cite{HPR19} to introduce the polymer model and related tools.
For a complete introduction to this model, see this wonderful book \cite{friedli_velenik_2017}.
Let $G$ be a graph and $\Omega$ be a finite set.
A polymer $\gamma=(\ol \gamma, \conf)$ consists of a support $\ol \gamma$ which is a connected subgraph of $G$ and a mapping $\conf$ which assigns to each vertex in $\ol\gamma$ some value in $\Omega$.
We use $\aabs{\gamma}$ to denote the number of vertices of $\ol\gamma$.
There is also a weight function $w(\gamma, \cdot): \=C \to \=C$ for each polymer $\gamma$.
There can be many polymers defined on the graph $G$ and we use $\Gamma^*=\Gamma^*(G)$ to denote the set of all polymers defined on it.
However, at the moment we do not give a constructive definition of polymers.
Such definitions are presented when they are needed, see \Cref{sec:polymerOfIS} and \Cref{sec:polymerOfColorings}.
We say two polymers $\gamma_1$ and $\gamma_2$ are compatible if $d_G(\ol{\gamma_1}, \ol{\gamma_2}) > 1$ and we use $\gamma_1\sim \gamma_2$ to denote that they are compatible.
For a subset $\Gamma\subseteq\Gamma^*$ of polymers, it is compatible if any two different polymers in this set are compatible.
We define $\+S(\Gamma^*) = \set{\Gamma\subseteq \Gamma^*\cmid \text{$\Gamma$ is compatible}}$ to be the collection of all compatible subsets of polymers.
For any $\Gamma\in\+S(\Gamma^*)$, we define $\ol \Gamma$ to be the the subgraph of $G$ by putting together the support of all polymers in $\Gamma$.
It is well defined since $\Gamma$ is compatible.
We also define $\aabs{\Gamma}$ to be the number of vertices of the subgraph $\ol\Gamma$ and $\conf[\Gamma]=\cup_{\gamma\in\Gamma}\conf$.
We say $(\Gamma^*,w)$ is a polymer model defined on the graph $G$ and the partition function of this polymer model is
\begin{align*}
  \Xi(G, z) = \sum_{\Gamma\in \+S(\Gamma^*)} \prod_{\gamma\in\Gamma} w(\gamma,z),
\end{align*}
where $z$ is a complex variable and $\prod_{\gamma\in\emptyset} w(\gamma,z)=1$ by convention. The following theorem\footnote{Here we only need a special case of the original theorem.} states conditions that $\Xi(G, z)$ can be approximated efficiently.

\begin{theorem}[\cite{HPR19}, Theorem 2.2] \label{thm:alg}
  Fix $\Delta$ and let $\+G$ be a set of graphs of degree at most $\Delta$. Suppose:
  \begin{itemize}
    \item There is a constant $C$ such that for all $G\in \+G$, the degree of $\Xi(G, z)$ is at most $C\abs{G}$.
    \item For all $G\in\+G$ and $\gamma\in\Gamma^*(G)$, $w(\gamma,z)=a_\gamma z^{\aabs \gamma}$ where $a_\gamma\neq 0$ can be computed in time $\exp(O(\aabs{\gamma} + \log_2 \abs{G}))$.
    \item For every connected subgraph $G'$ of every $G\in \+G$, we can list all polymers $\gamma\in\Gamma^*(G)$ with $\ol \gamma = G'$ in time $\exp(O(\abs{G'}))$.
    \item There is a constant $R > 0$ such that for all $G\in\+G$ and $z\in\=C$ with $\abs{z}<R$, $\Xi(G, z)\neq 0$.
  \end{itemize}
  Then for every $z$ with $\abs{z} < R$, there is an FPTAS for $\Xi(G, z)$ for all $G\in \+G$.
\end{theorem}

The following condition by Koteck\`y and Preiss (KP-condition) is useful to show that $\Xi(G, z)$ is zero-free in certain regions.

\begin{lemma}[\cite{KP86}] \label{lem:KP-condition}
  Suppose there is a function $a:\Gamma^*\to \=R_{> 0}$ and for every $\gamma^*\in\Gamma^*$,
  \begin{align*}
    \sum_{\gamma\cmid\gamma\not\sim \gamma^*} e^{a(\gamma)}\abs{w(\gamma,z)} \le a(\gamma^*).
  \end{align*}
  Then $\Xi(G, z)\neq 0$.
\end{lemma}

To verify the KP-condition, usually we need to enumerate polymers and the following lemma is useful to bound the number of enumerated polymers.

\begin{lemma}[\cite{DBLP:journals/rsa/BorgsCKL13}]\label{lem:EnumSubgraph}
  For any graph $G=(V,E)$ with maximum degree $\Delta$ and $v\in V$, the number of connected induced subgraphs of size $k\ge 2$ containing $v$ is at most $(e\Delta)^{k-1}/2$.
  As a corollary, the number of connected induced subgraphs of size $k\ge 1$ containing $v$ is at most $(e\Delta)^{k-1}$.
\end{lemma}

\subsection{Some useful lemmas}

Throughout this paper, we use $H(x)$ to denote the binary entropy function
\begin{align*}
    H(x)= -x\log_2 x - (1-x)\log_2 (1-x),\quad x\in(0,1).
\end{align*}
Moreover, we extend this function to the interval $[0,1]$ by defining $H(0)=H(1)=0$. This is reasonable since $\lim_{x\to 0^+}H(x)=\lim_{x\to 1^-}H(x)=0$.

\begin{lemma}\label{lem:polyUpperBoundOfEntropy}
  It holds that $H(x)\le 2\sqrt{x(1-x)}\le 2\sqrt{x}$ for all $0 \le x \le 1$.
\end{lemma}

\begin{proof}
  Let $f(x) = \frac{2\sqrt{x(1-x)}}{H(x)}$.
  Since $f(x) = f(1-x)$ and $f(1/2)=1$, it suffices to show that $\partial f / \partial x \ge 0$ for any $1/2 \le x < 1$.
  It holds that
  \begin{align*}
    \frac{\partial f}{\partial x}
    = \frac{(1-x)\log_2 1/(1-x) - x\log_2 1/x}{H(x)^2 \sqrt{x(1-x)}}
    \triangleq \frac{g(x)}{H(x)^2\sqrt{x(1-x)}}
    \ge 0
  \end{align*}
  for all $1/2 \le x < 1$, since $g(1/2)=0,\lim_{x\to 1^-} g(x)=0$ and $g$ is concave over $[1/2, 1)$.
  The concavity of $g$ follows from
  \begin{align*}
    \frac{\partial^2 g}{\partial x^2} = \frac{(1-2x)\log_2 e}{(1-x)x} \le 0
  \end{align*}
  for $1/2\le x < 1$.
\end{proof}

\begin{lemma}\label{lem:nonPolyUpperBoundOfEntropy}
  It holds that $H(x)\le -2x\log_2 x$ for all $0 < x \le 1/2$.
\end{lemma}
\begin{proof}
  Let $f(x)=H(x)+2x\log_2 x = x\log_2 x - (1-x)\log_2 (1-x)$, it suffices to show that $f(x) \le 0$ for $x\in(0,1/2]$. In fact, $\lim_{x\to 0^+} f(x)=0, f(1/2)=0$ and $f$ is convex over $(0,1/2]$. The convexity of $f$ follows from
  \begin{align*}
  	\frac{\partial^2 f}{\partial x^2} = \frac{(1-2x)\log_2 e}{x(1-x)} \ge 0
  \end{align*}
  for $0 < x \le 1/2$.
\end{proof}

\begin{lemma}\label{lem:monoOfEntropyVariant}
  For all $a \ge 1$, $H(x) - 1 / a H(ax) \ge x(\ln a - x)\log_2 e$ for all $0 \le x \le 1/a$.
\end{lemma}

\begin{proof}
  Recall that $\frac{-x}{1-x} \le \ln (1-x) \le -x$ for any $0 < x < 1$. Thus for any $0 < x < 1/a$,
  \begin{align*}
    H(x) - 1/a H(ax)
    &= \tp{x\ln a  - (1-x)\ln (1-x) + 1/a(1-ax)\ln(1-ax)}\cdot\log_2 e \\
    &\ge \tp{x\ln a - (1-x)(-x) + 1/a (1-ax) (-ax)/(1-ax)}\cdot \log_2 e \\
    &= x(\ln a -x)\cdot \log_2 e
  \end{align*}
And the inequality holds trivially for $x=0$ and $x=1/a$.
\end{proof}

\begin{lemma}\label{lem:upperBoundOfAVariantOfEntropy}
  It holds that $\displaystyle H\tp{\frac{x}{1-y}}(1-y) - H(x) \le -xy\log_2 e$ for all $0 \le x,y< 1$ with $x+y< 1$.
\end{lemma}

\begin{proof}
  It holds that for any $0 \le x, y < 1$ with $x+y<1$,
  \begin{align*}
    &~\phantom{=}~H\tp{x/(1-y)}(1-y) - H(x) + xy\log_2e \\
    &= \tp{(1-x)\ln (1-x) + (1-y)\ln (1-y) - (1-x-y)\ln (1-x-y) + xy}\log_2 e \\
    &\triangleq f(x,y)\log_2 e.
  \end{align*}
  Thus it suffices to show that $f(x,y)\le 0$ for $0 \le x, y < 1$ with $x + y < 1$.
  Fix $0 \le x < 1$.
  We verify that $f(x,0) = 0$ and
  \begin{align*}
    \frac{\partial f}{\partial y}
    = - \ln (1-y) + \ln (1-x-y) + x
    = \ln \frac{1-x-y}{1-y} + x
    \le -x / (1-y) + x
    \le 0
  \end{align*}
  for any $0 \le y < 1-x$.
\end{proof}

\begin{lemma}[{\cite[Lemma 10.2]{probandcompbook}}]\label{lem:boundsOfBinomCoefficients}
  Suppose that $n$ is a positive integer and $k\in [0,1]$ is a number such that $kn$ is an integer. Then
  \begin{align*}
    \frac{2^{H(k)n}}{n+1} \le \binom{n}{kn} \le 2^{H(k)n}.
  \end{align*}
\end{lemma}

\begin{lemma} \label{lem:coreMonotoneFunc}
  For $b>a>0$, the function $f(\lambda) = \lambda^a / (\lambda+1)^b$ is monotonically increasing on $[0,\frac{a}{b-a}]$ and monotonically decreasing on $[\frac{a}{b-a},+\infty)$.
\end{lemma}

\begin{proof}
    It holds that
    \begin{align*}
        \frac{\partial f}{\partial \lambda} = e^{\ln f(\lambda)} \cdot \frac{a-(b-a)\lambda}{\lambda(\lambda+1)}
    \end{align*}
    for all $\lambda>0$.
\end{proof}

\section{Counting independent sets for $\lambda \ge 1$} \label{sec:highFugacity}

Throughout this section, we consider integers $\Delta \ge \lowerBoundOfDeltaForIS$, fugacity $\lambda\ge 1$ and set parameters $\zeta,\alpha,\beta$ to be
\begin{align*}
  \zeta=\valueOfZeta,\alpha=\frac{\tOfIS}{\Delta},\beta=\frac{\Delta}{\tOfIS\zeta}.
\end{align*}
\begin{lemma} \label{lem:almostAtHighFugacity}
  For $\Delta \ge \lowerBoundOfDeltaForIS$, $\displaystyle\lim_{n\to\infty}\Pr_{G\sim \setG}\sqtp{G\in \graphClassAtHighFugacity}=1$.
\end{lemma}
\begin{proof}
  We verify that the conditions in \Cref{lem:conditionOfExpanders} are satisfied.
  Recall that $\zeta = \valueOfZeta, \alpha = \tOfIS / \Delta, \beta = \Delta / (\tOfIS\zeta)$ and $\Delta \ge \lowerBoundOfDeltaForIS$.
  Clearly $0 < \alpha < 1/\beta < 1$.
  Let $f(\Delta)= \Delta-\frac{H(\alpha) + H(\alpha \beta)}{H(\alpha) - \alpha \beta H(1/\beta)}$.
  It follows from \Cref{lem:monoOfEntropyVariant} that
  \begin{align*}
    H(\alpha) - \alpha \beta H(1/\beta)
    = H(\tOfIS/\Delta) - 1/\zeta H(\tOfIS\zeta / \Delta)
    &\ge \tOfIS/\Delta(\ln \zeta - \tOfIS/\Delta)\log_2 e \\
    &\ge \tOfIS/\Delta(\ln \valueOfZeta - \tOfIS / 1000)\log_2 e\\
    &\ge 1/\Delta
  \end{align*}
  for any $\Delta\ge 1000$.
  Then
  \begin{align*}
    f(\Delta)
    \ge \Delta - \frac{H(\tOfIS/1000) + H(1/\zeta)}{1/\Delta}
    \ge 0.2\Delta > 0
  \end{align*}
  for $\Delta \ge 1000$.
  For $\lowerBoundOfDeltaForIS\le \Delta < 1000$, we can use computers to verify that $f(\Delta) > 0$.
  Actually, in the current setting of parameters, $f(52) \approx -0.06 < 0 < f(\lowerBoundOfDeltaForIS)\approx 0.11$.
\end{proof}

In the rest of this section, whenever possible, we will simplify notations by omitting superscripts, subscripts and brackets with the symbols between (but this will not happen in the statement of lemmas and theorems).
For example, $Z(G, \lambda)$ may be written as $Z$ if $G$ and $\lambda$ are clear from context.

\subsection{Approximating $Z(G, \lambda)$}
For all $G=(\+L,\+R,E)\in\graphClassAtHighFugacity,\+X\in\set{\+L,\+R}$ and $\lambda\ge 1$, we define
\begin{align*}
  \clusterOfIS &= \set{I\in \+I(G)\cmid \abs{I\cap \+X}< \alpha n},
  Z_{\+X}(G,\lambda) = \sum_{I\in \clusterOfIS} \lambda^{\abs I}.
\end{align*}
The main result in this part is that we can use $Z_\+L(G,\lambda) + Z_\+R(G,\lambda)$ to approximate $Z(G, \lambda)$.

\begin{lemma}\label{lem:strucApproxOfIS}
  For $\Delta\ge\lowerBoundOfDeltaForIS$ and $\lambda \ge 1$, there are constants $C=C(\Delta)>1$ and $N=N(\Delta)$ so that for all  $G\in\graphClassAtHighFugacity$ with $n > N$ vertices on both sides, $Z_\+L(G,\lambda)+Z_\+R(G,\lambda)$ is a $C^{-n}$-relative approximation to $Z(G, \lambda)$.
\end{lemma}

\begin{proof}
  Let $N_1,C_1,N_2,C_2$ be the constants in \Cref{lem:strucApprox1OfIS} and \Cref{lem:strucApprox2OfIS}, respectively.
  It follows from these lemmas that
  \begin{align*}
    \exp(-(C_1^{-n} + C_2^{-n})) Z \le Z_\+L + Z_\+R \le \exp(C_1^{-n} + C_2^{-n}) Z
  \end{align*}
  for all $n > \max(N_1,N_2)$.
  It is clear that $C_1^{-n} + C_2^{-n} \le 2\min(C_1,C_2)^{-n} = \tp{\min(C_1,C_2)/2^{1/n}}^{-n} < C^{-n}$ for another constant $C=C(\Delta)>1$ and for all $n > N\ge \max(N_1,N_2)$ where $N=N(\Delta)$ is another sufficiently large constant.
  Therefore we obtain
  \begin{align*}
    \exp(-C^{-n}) Z \le Z_\+L + Z_\+R \le \exp(C^{-n}) Z
  \end{align*}
  for all $n > N$.
\end{proof}

\begin{lemma} \label{lem:strucApprox1OfIS}
    For $\Delta\ge 3$ and $\lambda \ge 1$, there are constants $C=C(\Delta)>1$ and $N=N(\Delta)$ so that for all  $G\in\graphClassAtHighFugacity$ with $n > N$ vertices on both sides, $\sum_{I\in \+I_\+L(G)\cup \+I_\+R(G)}\lambda^{\abs{I}}$ is a $C^{-n}$-relative approximation to $Z(G,\lambda)$.
\end{lemma}

\begin{proof}
  It is clear that
  \begin{align}
    \label{eq:midstep7}
    \sum_{I\in \+I_\+L\cup \+I_\+R} \lambda^{\abs I}
    \ge (\lambda+1)^n.
  \end{align}
  Let $\+B = \+I \setminus \tp{\+I_\+L\cup \+I_\+R}$.
  For any $I\in \+B$, it follows from the definition of $\+B$ that $\abs{I\cap \+L}\ge \alpha n$ and $\abs{I\cap \+R}\ge \alpha n$.
  Using the expansion property, we obtain $\abs{N(I\cap \+L)}\ge \beta\lfloor \alpha n\rfloor$ and thus $\abs{I\cap \+R}\le n - \abs{N(I\cap \+L)}\le (1-1/\zeta)n$ where $1/\zeta = \beta\lfloor \alpha n\rfloor/n \geq \alpha\beta - \beta/n$.
  Analogously, it holds that $\abs{I\cap \+L}\le (1-1/\zeta)n$.
  In the following, we assume $n\ge N_1$ for some $N_1=N_1(\Delta)>0$, such that
  \begin{align}
    1-1/\zeta\le 0.219.
  \end{align}
  We obtain an upper bound of $\sum_{I\in \+B}\lambda^{\abs I}$ as follows:
  \begin{enumerate}[label=\alph*)]
    \item Consider an independent set $I\in \+B$.
      Recall that $\alpha n \le \abs{I\cap \+L} \le (1-1/\zeta)n$.
      We first enumerate a subset $U\subseteq \+L$ with $\alpha n \le \abs U \le (1-1/\zeta)n$ and then enumerate all independent sets $I$ with $I\cap \+L = U$.
      Since $1-1/\zeta<1/2$, there are at most
      \begin{align*}
        n\binom{n}{\floor{(1-1/\zeta)n}}
        \le n2^{H(1-1/\zeta)n}
      \end{align*}
      ways to enumerate such a set $U$, where the inequality follows from \Cref{lem:boundsOfBinomCoefficients}.
    \item Now fix a set $U\subseteq \+ L$.
      Recall that every independent set $I\in \+B$ satisfies $\abs{I\cap \+R}\le (1-1/\zeta)n$.
      Therefore
      \begin{align*}
        \sum_{I\in\+B\cmid \abs{I\cap \+L}=U} \lambda^{\abs{I}}
        = \lambda^{\abs{U}} \sum_{I\in\+B\cmid \abs{I\cap \+L}=U}  \lambda^{\abs{I\cap \+R}}
        \le \lambda^{(1-1/\zeta)n}\tp{\lambda + 1}^{(1-1/\zeta) n}.
      \end{align*}
    \item Combining the first two steps we obtain
    \begin{align}
      \label{eq:midstep8}
      \sum_{I\in\+B}\lambda^{\abs I} \le n 2^{H(1-1/\zeta)n}\lambda^{(1-1/\zeta)n} (\lambda+1)^{(1-1/\zeta)n}
      = n 2^{H(1-1/\zeta)n}(\lambda^2+\lambda)^{(1-1/\zeta)n}.
    \end{align}
  \end{enumerate}
  Using \Cref{eq:midstep7} and \Cref{eq:midstep8}, we obtain
  \begin{align}
    \label{eq:midstep1}
    \frac{\sum_{I\in \+B}\lambda^{\abs I}}{\sum_{I\in \+I_\+L\cup \+I_\+R} \lambda^{\abs I}}
    \le
    \frac{n2^{H(1-1/\zeta)n}(\lambda^2+\lambda)^{(1-1/\zeta)n}}{(\lambda+1)^n}
    = n(f(\lambda))^n,
  \end{align}
  where
  \begin{align*}
    f(\lambda)=2^{H(1-1/\zeta)}\cdot \frac{\lambda^{1-1/\zeta}}{(\lambda+1)^{1/\zeta}}.
  \end{align*}
  Since $1-1/\zeta < 1/\zeta$, it follows from \Cref{lem:coreMonotoneFunc} that
  \begin{align*}
    f(\lambda) \le f(1) =2^{H(1-1/\zeta)-1/\zeta} <1
  \end{align*}
  for all $\lambda\ge 1$. So there exists some constant $C>1$ such that
  \begin{align*}
    \text{\Cref{eq:midstep1}} \le n (f(1))^n < C ^{-n}
  \end{align*}
  for all $n > N \ge N_1$ where $N=N(\Delta)$ is another sufficiently large constant.
  Using the upper bound on \Cref{eq:midstep1} and $1+x \le \exp(x)$ for any $x\in \=R$ we obtain
  \begin{align*}
    \sum_{I\in \+I_\+L\cup \+I_\+R} \lambda^{\abs I}
    \le Z
    = \sum_{I\in \+I_\+L\cup \+I_\+R} \lambda^{\abs I} + \sum_{I\in \+B}\lambda^{\abs I}
    \le \exp\tp{C^{-n}}\sum_{I\in \+I_\+L\cup \+I_\+R} \lambda^{\abs I}
  \end{align*}
  for all $n > N$.
\end{proof}

\begin{lemma} \label{lem:strucApprox2OfIS}
  For $\Delta\ge\lowerBoundOfDeltaForIS$ and $\lambda \ge 1$, there are constants $C>1$ and $N$ so that for all  $G\in\graphClassAtHighFugacity$ with $n > N$ vertices on both sides, $Z_\+L(G,\lambda) + Z_\+R(G,\lambda)$ is a $C^{-n}$-relative approximation to $\sum_{I\in \+I_\+L(G)\cup \+I_\+R(G)}\lambda^\abs{I}$.
\end{lemma}

\begin{proof}
  For any $I\in \+I_\+L\cap \+I_\+R$, it holds that $\abs{I\cap \+L}< \alpha n$ and $\abs{I\cap \+R}< \alpha n$.
  Clearly $\sum_{I\in \+I_\+L\cup \+I_\+R}\lambda^\abs{I}\ge (\lambda+1)^n$.
  Therefore
  \begin{align}
    \label{eq:midstep6}
    \frac{\sum_{I\in \+I_\+L\cap \+I_\+R}\lambda^\abs{I}}{\sum_{I\in \+I_\+L\cup \+I_\+R}\lambda^\abs{I}}
    \le (\lambda+1)^{-n}\tp{\sum_{k=0}^{\floor{\alpha n}} \binom{n}{k} \lambda^k}^2
    \le n^2\tp{\frac{4^{H(\alpha)} \lambda^{2\alpha}}{\lambda+1}}^{n},
  \end{align}
  where the last inequality follows from \Cref{lem:boundsOfBinomCoefficients}.
  Recall that $\alpha = \tOfIS / \Delta$ and $\Delta \ge \lowerBoundOfDeltaForIS$.
  Then
  \begin{align*}
    \frac{4^{H(\alpha)}\lambda^{2\alpha}}{\lambda+1}\bigg|_{\lambda=1} \le 0.76 < 1.
  \end{align*}
  It follows from \Cref{lem:coreMonotoneFunc} that $4^{H(\alpha)}\lambda^{2\alpha} / (\lambda + 1)$ is monotonically decreasing in $\lambda$ on $[1,\+\infty)$ for all fixed $\Delta \ge \lowerBoundOfDeltaForIS$.
  Thus
  \begin{align*}
    \text{\Cref{eq:midstep6}}
    \le \tp{1/\tp{0.76n^{2/n}}}^{-n} < C^{-n}
  \end{align*}
  for some constant $C>1$ and for all $n > N$ where $N$ is a sufficiently large constant.
  Using the upper bound on \Cref{eq:midstep6} and $1+x\le \exp(x)$ for any $x\in \=R$ we obtain
  \begin{align*}
    \sum_{I\in \+I_\+L\cup \+I_\+R} \lambda^{\abs I}
    \le Z_\+L + Z_\+R
    =\sum_{I\in \+I_\+L\cup \+I_\+R} \lambda^{\abs I} + \sum_{I\in \+I_\+L\cap \+I_\+R} \lambda^{\abs I}
    \le \exp(C^{-n}) \sum_{I\in \+I_\+L\cup \+I_\+R} \lambda^{\abs I}
  \end{align*}
  for all $n > N$.
\end{proof}

\subsection{Approximating $Z_\+X(G,\lambda)$} \label{sec:polymerOfIS}

In this subsection, we discuss how to approximate $Z_\+X(G, \lambda)$ for any graph $G\in \graphClassAtHighFugacity, \+X\in \set{\+L,\+R}$ and $\lambda\ge 1$.
We will use the polymer model (see \Cref{sec:polymer-model}).
First we constructively define the polymers we need.
For any $I\in \+I_\+X(G)$, we can partition the graph $(G^2)[I\cap \+X]$ into connected components $U_1,U_2,\ldots,U_k$ for some $k\ge 0$ (trivially $k=0$ if $I\cap \+X=\emptyset$). There are no edges in $G^2$ between $U_i$ and $U_j$ for any $1\le i \neq j \le k$.
If $k > 0$, let $p(I) = \set{(U_1, \*1_{U_1}), (U_2, \*1_{U_2}), \ldots, (U_k, \*1_{U_k})}$ where $\*1_{U_i}$ is the unique mapping from $U_i$ to $\set{1}$.
If $k=0$, let $p(I)=\emptyset$.
We define the set of all polymers to be
\begin{align*}
  \Gamma^*_{\+X}(G) = \bigcup_{I\in \+I_\+X(G)} p(I)
\end{align*}
and each element in this set is called a polymer. When the graph $G$ and $\+X$ are clear from the context, we simply denote by $\Gamma^*$ the set of polymers.
Clearly, $p$ is a mapping from $\+I_\+X(G)$ to the set $\set{\Gamma\in\+S(\Gamma^*_{\+X}(G))\cmid \aabs{\Gamma} < \alpha n}$ since $\aabs{p(I)} = \abs{I\cap \+X} < \alpha n$ for all $I\in \+I_\+X(G)$.
For each polymer $\gamma$, define its weight function $w(\gamma,\cdot)$ as
\begin{align*}
  w(\gamma,z) = \lambda^{\aabs \gamma}(\lambda+1)^{-\abs{N(\ol\gamma)}} z^{\aabs \gamma},
\end{align*}
where $z$ is a complex variable. The weight function can be computed in polynomial time in $\abs{\ol \gamma}$.

The partition function of the polymer model $(\Gamma^*,w)$ on the graph $G^2$ is the following sum:
\[
    \Xi(z) = \sum_{\Gamma\in\+S(\Gamma^*)} \prod_{\gamma\in\Gamma} w(\gamma,z).
\]
Recall that two polymers $\gamma_1$ and $\gamma_2$ are compatible if $d_{G^2}(\ol{\gamma_1},\ol{\gamma_2})>1$ and this condition is equivalent to $d_G(\ol{\gamma_1}, \ol{\gamma_2}) > 2$.

\begin{lemma}\label{lem:exactRepOfIS}
  For all bipartite graphs $G=(\+L,\+R,E)$ with $n$ vertices on both sides, $\+X\in \set{\+L,\+R}$ and $\lambda\ge 0$,
  \begin{align*}
  Z_\+X(G,\lambda)=(\lambda+1)^n\sum_{\Gamma\in\+S(\Gamma^*_{\+X}(G))\cmid \aabs{\Gamma} < \alpha n} \prod_{\gamma\in\Gamma}w(\gamma,1).
  \end{align*}
\end{lemma}

\begin{proof}
  Recall that in the definition of polymers, $p$ is a mapping from $\+I_\+X$ to $\set{\Gamma\in\+S(\Gamma^*)\cmid \aabs{\Gamma}<\alpha n}$.
  Thus
  \begin{align*}
    Z_\+X(G,\lambda) = \sum_{I\in \+I_\+X}\lambda^\abs{I}
    = \sum_{\Gamma\in\+S(\Gamma^*)\cmid \aabs{\Gamma}<\alpha n} \sum_{I\in \+I_\+X\cmid p(I)=\Gamma} \lambda^{\abs I}.
  \end{align*}
  Fix $\Gamma\in\+S(\Gamma^*)$ with $\aabs{\Gamma} < \alpha n$.
  It holds that
  \begin{align}
    \label{eq:midstep9}
    \sum_{I\in \+I_\+X\cmid p(I)=\Gamma} \lambda^{\abs I}
    = \sum_{I\in \+I_\+X\cmid I\cap \+X = \ol\Gamma} \lambda^{\abs I}
    = \lambda^{\aabs{\Gamma}}(\lambda+1)^\abs{(\+L\sqcup\+R)\setminus (\+X\sqcup N_G(\ol \Gamma))},
  \end{align}
  where the last equality follows from $\aabs{\Gamma} < \alpha n$.
  Since $\Gamma$ is compatible, $N_G(\ol\Gamma)=\sqcup_{\gamma\in\Gamma}N_G(\ol\gamma)$ and $\abs{(\+L\sqcup\+R)\setminus (\+X\sqcup N_G(\ol\Gamma))}=n-\sum_{\gamma\in\Gamma}\abs{N_G(\ol\gamma)}$.
  Thus
  \begin{align*}
    \text{\Cref{eq:midstep9}}
    &= \lambda^{\sum_{\gamma\in\Gamma}\aabs\gamma} (\lambda+1)^{n-\sum_{\gamma\in\Gamma} N(\ol \gamma)} \\
    & = (\lambda+1)^n \prod_{\gamma\in\Gamma} \lambda^{\aabs\gamma} (\lambda+1)^{-\abs{N(\ol \gamma)}} \\
    &= (\lambda+1)^n \prod_{\gamma\in\Gamma} w(\gamma,1).
  \end{align*}
  This completes the proof.
\end{proof}

\begin{lemma} \label{lem:approxRepOfIS}
  For $\Delta\ge\lowerBoundOfDeltaForIS$ and $\lambda \ge 1$, there are constants $C>1$ and $N$ so that for all $G=(\+L,\+R, E)\in \graphClassAtHighFugacity$ with $n>N$ vertices on both sides and $\+X\in \set{\+L,\+R}$,
  \begin{align*}
    (\lambda+1)^n \Xi(1) = (\lambda+1)^n \sum_{\Gamma\in\+S(\Gamma^*_\+X(G))} \prod_{\gamma\in \Gamma} w(\gamma,1)
  \end{align*}
  is a $C^{-n}$-relative approximation to $Z_\+X(G,\lambda)$.
\end{lemma}

\begin{proof}
  It is clear that $Z_\+X(G,\lambda)\ge (\lambda+1)^n$.
  Then using \Cref{lem:exactRepOfIS} and \Cref{lem:decayRateOfIS} we obtain
  \begin{align}
    \label{eq:midstep5}
    \frac{(\lambda+1)^n \Xi(1) - Z_\+X(G,\lambda)}{Z_\+X(G,\lambda)}
    &\le \sum_{\Gamma\in\+S(\Gamma^*)\cmid \aabs\Gamma \ge \alpha n}\prod_{\gamma\in\Gamma} w(\gamma,1)
    \le \sum_{\Gamma\in\+S(\Gamma^*)\cmid \aabs\Gamma \ge \alpha n} 2^{-\beta \aabs\Gamma}.
  \end{align}
  To enumerate each $\Gamma\in\+S(\Gamma^*)$ with $\aabs\Gamma \ge \alpha n$ at least once, we first enumerate an integer $\alpha n \le k \le n$, then since $\ol\Gamma\subseteq \+X$, we choose $k$ vertices from $\+X$.
  Therefore
  \begin{align*}
    \text{\Cref{eq:midstep5}}
    \le \sum_{k=\ceil{\alpha n}}^{n} \binom{n}{k} 2^{-\beta k}
    \le \sum_{k=\ceil{\alpha n}}^{n} 2^{H(k/n)n} 2^{-\beta k}
    \le \sum_{k=\ceil{\alpha n}}^{n} \tp{2^{2\sqrt{n/k} - \beta}}^k
    \le \sum_{k=\ceil{\alpha n}}^{n} \tp{2^{2\sqrt{1/\alpha} - \beta}}^k,
  \end{align*}
  where the inequalities follow from \Cref{lem:boundsOfBinomCoefficients} and \Cref{lem:polyUpperBoundOfEntropy}.
  Recall that $\zeta=\valueOfZeta$, $\alpha = \tOfIS / \Delta, \beta = \Delta / (\tOfIS\zeta)$ and $\Delta\ge \lowerBoundOfDeltaForIS$.
  Let $f(\Delta) = 2\sqrt{1/\alpha} - \beta = 2\sqrt{\Delta/\tOfIS}-\Delta/(\tOfIS\zeta)$.
  We obtain
  \begin{align*}
    \text{\Cref{eq:midstep5}}
    \le \frac{2^{f(\Delta)\alpha n}}{1-2^{f(\Delta)}}
    &= \frac{\tp{2^{2\sqrt{\tOfIS/\Delta}-1/\zeta}}^n}{1-2^{f(\Delta)}}
  \end{align*}
  It follows from \Cref{lem:monotoneFunc2} that $f(\Delta)$ is monotonically decreasing in $\Delta$ on $[\lowerBoundOfDeltaForIS, +\infty)$.
  Thus
  \begin{align*}
    \text{\Cref{eq:midstep5}}
    \le
    \frac{\tp{2^{2\sqrt{\tOfIS/\lowerBoundOfDeltaForIS}-1/\valueOfZeta}}^n}{1-2^{2\sqrt{\lowerBoundOfDeltaForIS/\tOfIS}-\lowerBoundOfDeltaForIS/(\tOfIS\times \valueOfZeta)}}
    \le
    0.81^n/0.98
    < C^{-n}
  \end{align*}
  for some constant $C>1$ and for all $n > N$ where $N$ is a sufficiently large constant.
  Using the upper bound on \Cref{eq:midstep5} and $1+x\le \exp(x)$ for any $x\in\=R$ we obtain
  \begin{align*}
    Z_\+X(G,\lambda) &\le (\lambda+1)^n \Xi(1) \\
    &= Z_\+X(G,\lambda) + ((\lambda+1)^n \Xi(1) - Z_\+X(G,\lambda)) \\
    & \le \exp(C^{-n}) Z_\+X(G,\lambda)
  \end{align*}
  for all $n > N$.
\end{proof}

\begin{lemma}\label{lem:monotoneFunc2}
  The function $f(\Delta)=2\sqrt{1/\alpha} - \beta$ is monotonically decreasing on $[\lowerBoundOfDeltaForIS, +\infty)$.
\end{lemma}

\begin{proof}
  Recall that $\zeta=\valueOfZeta$, $\alpha = \tOfIS / \Delta, \beta = \Delta / (\tOfIS\zeta)$. It holds that
  \begin{align*}
    \frac{\partial f}{\partial \Delta} = \frac{1}{\sqrt{\tOfIS\Delta}} - \frac{1}{\tOfIS\zeta} \le\frac{1}{\sqrt{\tOfIS\times \lowerBoundOfDeltaForIS}} - \frac{1}{\tOfIS\times \valueOfZeta} \approx -0.19 < 0
  \end{align*}
  for all $\Delta \ge \lowerBoundOfDeltaForIS$.
\end{proof}

\begin{lemma}\label{lem:decayRateOfIS}
  For all polymers $\gamma\in \Gamma^*$ defined by $G=(\+L,\+R,E)\in \graphClassAtHighFugacity$, $\+X\in\{\+L,\+R\}$ and $\lambda \ge 1$,
  \begin{align*}
    \abs{w(\gamma,z)} \le (2^{-\beta}\abs{z})^{\aabs\gamma}.
  \end{align*}
  As a corollary, $w(\gamma,1) \le 2^{-\beta\aabs\gamma}$ and for all compatible $\Gamma\subseteq \Gamma^*(G)$,
  \begin{align*}
    \prod_{\gamma\in \Gamma}w(\gamma, 1) \le 2^{-\beta\aabs\Gamma}.
  \end{align*}
\end{lemma}

\begin{proof}
  Let $n=\abs{\+L}=\abs{\+R}$ and let $\gamma$ be any polymer. It follows from the definition of polymers that $\aabs\gamma \le \alpha n$ and by the expansion property, $\abs{N(\ol \gamma)} \ge \beta \aabs{\gamma}$.
  Thus we have
  \begin{align*}
    \abs{w(\gamma,z)} &= \lambda^{\aabs\gamma}(\lambda+1)^{-\abs{N(\ol\gamma)}}\abs{z}^{\aabs\gamma} \\
    &\le (\lambda(\lambda+1)^{-\beta})^{\aabs \gamma}\abs{z}^{\aabs\gamma}
    \le (2^{-\beta}\abs{z})^{\aabs\gamma}
  \end{align*}
  where the last inequality follows from \Cref{lem:coreMonotoneFunc} since $\beta>1$ and $\lambda \ge 1$. In particular, $w(\gamma,1) \le 2^{-\beta\aabs\gamma}$. For any compatible $\Gamma$, it holds that $\aabs{\Gamma} = \sum_{\gamma\in\Gamma} \aabs{\gamma}$. Thus
  $
    \prod_{\gamma\in \Gamma}w(\gamma, 1) \le
    \prod_{\gamma\in \Gamma}2^{-\beta \aabs \gamma} =
    2^{-\beta\aabs\Gamma}$.
\end{proof}

\subsection{Approximating the partition function of the polymer model}

\begin{lemma} \label{lem:algOfISPolymerPF}
  For $\Delta\ge \lowerBoundOfDeltaForIS$ and $\lambda\ge 1$, there is an FPTAS for $\Xi(1)$ for all $G=(\+L,\+R,E)\in\graphClassAtHighFugacity$ and $\+X\in\set{\+L,\+R}$.
\end{lemma}

\begin{proof}
  We use the FPTAS in \Cref{thm:alg} to design the FPTAS we need.
  To this end, we generate a graph $G^2$ in polynomial time in $\abs{G}$ for any $G\in\graphClassAtHighFugacity$.
  We use this new graph $G^2$ as input to the FPTAS in \Cref{thm:alg}.
  It is straightforward to verify the first three conditions in \Cref{thm:alg}, only with the exception that the information of $G^2$ may not be enough because certain connectivity information in $G$ is discarded in $G^2$.
  Nevertheless, we can use the original graph $G$ whenever needed and thus the first three conditions are satisfied.
  For the last condition, \Cref{lem:zeroFreeOfISPolymerPF} verifies it.
\end{proof}

\begin{lemma} \label{lem:zeroFreeOfISPolymerPF}
  There is a constant $R>1$ so that for $\Delta\ge\lowerBoundOfDeltaForIS$ and $\lambda\ge 1$, $\Xi(z)\neq 0$ for all $G\in\graphClassAtHighFugacity$, $\+X\in\set{\+L,\+R}$ and $z\in\=C$ with $\abs{z} <R$, .
\end{lemma}

\begin{proof}
  Set $R=1.001$.
  For any $\gamma\in\Gamma^*$, let $a(\gamma)=t\aabs{\gamma}$ where $t=\tp{-1+\sqrt{1+8e}}/(4e)\approx 0.346$.
  We will verify that the KP-condition
  \begin{align}
    \sum_{\gamma:\gamma\not\sim\gamma^*} e^{t\aabs\gamma} \abs{w(\gamma,z)} \le t\aabs{\gamma^*} \label{eq:ISKPCondition1}
  \end{align}
  holds for any $\gamma^*\in \Gamma^*$ and any $\abs{z} < R$. It then follows from \Cref{lem:KP-condition} that $\Xi(z)\neq 0$ for any $\abs{z} < R$.
  Recall that $d_{G^2}(\ol\gamma, \ol{\gamma^*})\le 1$ for all $\gamma\not\sim\gamma^*$.
  Thus there is always a vertex $v\in \ol\gamma\subseteq \+X$ such that $v\in \ol{\gamma^*}\sqcup N_{G^2}(\ol{\gamma^*})$.
  The number of such vertices $v$ is at most $\Delta^2\aabs{\gamma^*}$.
  So to enumerate each $\gamma\not\sim\gamma^*$ at least once, we can
  \begin{enumerate}[label=\alph*)]
    \item first enumerate a vertex $v$ in $\+X\cap \tp{\ol{\gamma^*}\cup N_{G^2}(\ol{\gamma^*})}$;
    \item then enumerate an integer $k$ from $1$ to $\floor{\alpha n}$;
    \item finally enumerate $\gamma$ with $v\in \ol\gamma$ and $\aabs{\gamma}=k$.
  \end{enumerate}
  Since $\ol\gamma$ is connected in $G^2$, applying \Cref{lem:EnumSubgraph} and using \Cref{lem:decayRateOfIS} to bound $\abs{w(\gamma,z)}$ we obtain
    \begin{align}
    \label{eq:midstep2}
    \sum_{\gamma:\gamma\not\sim\gamma^*}e^{t\aabs\gamma} \abs{w(\gamma,z)}
    &\le \Delta^2\aabs{\gamma^*} \tp{
        e^t 2^{-\beta}\abs{z} +
        \sum_{k=2}^{\floor{\alpha n}}\tp{e\Delta^2}^{k-1}2^{-1} e^{tk} 2^{-\beta k} \abs{z}^k}.
  \end{align}
  Let $x=e^{t+1}\Delta^2 2^{-\beta} R$. Since $\abs{z}<R$, we obtain
  \begin{align*}
    \sum_{\gamma:\gamma\not\sim\gamma^*}e^{t\aabs\gamma} \abs{w(\gamma,z)}
    \le \frac{x}{e} \aabs{\gamma^*} \tp{1+\frac{1}{2}\sum_{k=2}^{\infty} x^{k-1}}
    = \frac{x(2-x)}{2e(1-x)} \cdot \aabs{\gamma^*}.
  \end{align*}
  Recall that $\zeta=\valueOfZeta$, $\beta = \Delta / (\tOfIS\zeta)$ and $\Delta \ge \lowerBoundOfDeltaForIS$.
  It follows from \Cref{lem:monotoneFunc1} that $\Delta^2 2^{-\beta}$ is monotonically decreasing in $\Delta$ on $[\lowerBoundOfDeltaForIS,+\infty)$.
  Thus it holds that
  \begin{align*}
    x=e^{t+1}\Delta^2 2^{-\beta} R
    \le \tp{e^{t+1}\Delta^2 2^{-\beta} R}\bigg|_{\Delta=\lowerBoundOfDeltaForIS} \le 0.545,
  \end{align*}
  and hence
  \begin{align*}
    \frac{x(2-x)}{2e(1-x)} < 0.33 <t.
  \end{align*}
  This completes the proof.
\end{proof}

\begin{lemma}\label{lem:monotoneFunc1}
  The function $f(\Delta)=\Delta^2 2^{-\beta}$ is monotonically decreasing on $[\lowerBoundOfDeltaForIS, +\infty)$.
\end{lemma}

\begin{proof}
  Recall that $\zeta=\valueOfZeta$, $\beta = \Delta / (\tOfIS\zeta)$.
  It is equivalent to show that $\partial \ln f / \partial \Delta < 0$ for all $\Delta\ge \lowerBoundOfDeltaForIS$.
  It holds that
  \begin{align*}
    \frac{\partial \ln f}{\partial \Delta} = \frac{2}{\Delta} - \frac{\ln 2}{\tOfIS\zeta} \le \frac{2}{\lowerBoundOfDeltaForIS} - \frac{\ln 2}{\tOfIS\times \valueOfZeta} \approx -0.15 < 0
  \end{align*}
  for all $\Delta\ge \lowerBoundOfDeltaForIS$.
\end{proof}

\subsection{Putting things together}

Using the results from previous parts, we obtain our main result for counting independent sets.

{\renewcommand{\thetheorem}{\ref{thm:BIS}}
\begin{theorem}\label{thm:mainAtHighFugacity}
  For $\Delta\ge 53$ and fugacity $\lambda\ge 1$, with high probability (tending to $1$ as $n\to \infty$) for a graph $G$ chosen uniformly at random from $\setG$, there is an FPTAS for the partition function $Z(G,\lambda)$.
\end{theorem}
\addtocounter{theorem}{-1}}

\begin{proof}
  This theorem follows from \Cref{lem:almostAtHighFugacity} and \Cref{lem:algOfIS}.
\end{proof}

\begin{algorithm}[htbp]
  \caption{Counting independent sets at fugacity $\lambda\ge 1$ for $\Delta\ge\lowerBoundOfDeltaForIS$}
  \label{alg:IS}
  \begin{algorithmic}[1]
    \State \textbf{Input:} \emph{A graph $G=(\+L,\+R,E)\in \graphClassAtHighFugacity$ with $n$ vertices on both sides and $\eps>0$}
    \State \textbf{Output:} \emph{$\widehat Z$ such that $\exp(-\eps)\widehat Z \le Z(G, \lambda)\le \exp(\eps)\widehat Z$}
    \If {$n\le N$ or $\eps \le 2C^{-n}$}
      \State Use the brute-force algorithm to compute $\widehat Z \gets Z(G, \lambda)$;
      \State Exit;
    \EndIf
    \State $\eps'\gets \eps - C^{-n}$;
    \State Use the FPTAS in \Cref{lem:algOfISPolymerPF} to obtain $\widehat{Z}_\+L$, an $\eps'$-relative approximation to the partition function $\Xi(z)$ at $z=1$ of the polymer model $(\Gamma^*_\+L(G),w)$.
    \State Use the FPTAS in \Cref{lem:algOfISPolymerPF} to obtain $\widehat{Z}_\+R$, an $\eps'$-relative approximation to the partition function $\Xi(z)$ at $z=1$ of the polymer model $(\Gamma^*_\+R(G),w)$.
    \State $\widehat Z\gets (\lambda+1)^n\tp{\widehat{Z}_\+L+\widehat{Z}_\+R}$;
  \end{algorithmic}
\end{algorithm}

\begin{lemma} \label{lem:algOfIS}
  For $\Delta\ge\lowerBoundOfDeltaForIS$ and $\lambda\ge 1$, there is an FPTAS for $Z(G, \lambda)$ for all $G\in\graphClassAtHighFugacity$.
\end{lemma}

\begin{proof}
  First we state our algorithm.
  See \Cref{alg:IS} for a pseudocode description.
  The input is a graph $G=(\+L,\+R,E)\in \graphClassAtHighFugacity$ and an approximation parameter $\eps > 0$.
  The output is a number $\widehat Z$ to approximate $Z(G, \lambda)$. We use $\Xi_{\+ X}(z)$ to denote the partition function of the polymer model $(\Gamma^*_\+X(G),w)$ for $\+X\in\{\+L,\+R\}$.
  Let $N_1,C_2,N_2,C_2$ be the constants in \Cref{lem:strucApproxOfIS} and \Cref{lem:approxRepOfIS}, respectively.
  These two lemmas show that $(\lambda+1)^n \tp{\Xi_\+L(1) + \Xi_\+R(1)}$ is a $C_1^{-n}+C_2^{-n}\le 2\min(C_1,C_2)^{-n}\le C^{-n}$-relative approximation to $Z(G, \lambda)$ for another constant $C > 1$ and all $n > N \ge \max(N_1,N_2)$ where $N$ is another sufficiently large constant.
  If $n\le N$ or $\eps \le 2C^{-n}$, we use the brute-force algorithm to compute $Z(G, \lambda)$.
  If $\eps > 2C^{-n}$, we apply the FPTAS in \Cref{lem:algOfISPolymerPF} with approximation parameter $\eps' = \eps-C^{-n}$ to obtain outputs $\widehat{Z}_\+L$ and $\widehat{Z}_\+R$ which approximate $\Xi_\+L(1)$ and $\Xi_\+R(1)$ , respectively.
  Let $\widehat Z = (\lambda+1)^n(\widehat{Z}_\+L + \widehat{Z}_\+R)$ be the output.
  It is clear that $\exp(-\eps) \widehat Z \le Z(G, \lambda)\le \exp(\eps)\widehat Z$.

  Then we show that \Cref{alg:IS} is indeed an FPTAS.
  It is required that the running time of our algorithm is bounded by $\tp{n/\eps}^{C_3}$ for some constant $C_3$ and for all $n > N_3$ where $N_3$ is a constant.
  Let $N_3=N$.
  If $\eps \le 2C^{-n}$, the running time of the algorithm would be $2.1^n \le (nC^n/2)^{C_3} \le \tp{n/\eps}^{C_3}$ for sufficient large $C_3$.
  If $\eps > 2C^{-n}$, the running time of the algorithm would be $\tp{n/\eps'}^{C_4}=\tp{n/(\eps - C^{-n})}^{C_4}\le \tp{2n/\eps}^{C_4} \le \tp{n/\eps}^{C_3}$ for sufficient large $C_3$, where $C_4$ is a constant from the FPTAS in \Cref{lem:algOfISPolymerPF}.
\end{proof}

\section{Counting independent sets for $\lambda=\lowerBoundOfLambdaInOmegaForm$}

Let $\lambda_l =\valueOfLambdaL=\lowerBoundOfLambdaInOmegaForm$.
Throughout this section, we consider sufficiently large integers $\Delta$, fugacity $\lambda > \lambda_l$ and set parameters $\alpha,\beta$ to be
\begin{align*}
  \alpha=\valueOfAlphaAtLowFugacity,
  \beta=\valueOfBetaAtLowFugacity.
\end{align*}
We define a set $\graphClassAtLowFugacity$ of graphs as
\begin{align*}
  \graphClassAtLowFugacity = \set{G\in \graphClassAtHighFugacity\cmid \text{$G$ has the $(\alpha,\alpha)$-cover property}}.
\end{align*}

\begin{lemma} \label{lem:almostAtLowFugacity}
  For all sufficiently large integers $\Delta$,
  $\displaystyle\lim_{n\to\infty}\Pr_{G\sim \setG}\sqtp{G\in \graphClassAtLowFugacity}=1$.
\end{lemma}

\begin{proof}
  In this proof we only consider sufficiently large integers $\Delta$.
  Recall that $\alpha = \valueOfAlphaAtLowFugacity$ and $\beta = \valueOfBetaAtLowFugacity$.
  It suffices to show that
  \begin{align}
    \label{eq:midstep21}
    \lim_{n\to\infty}\Pr_{G\sim \setG}\sqtp{G\in \graphClassAtHighFugacity} &= 1, \\
    \label{eq:midstep22}
    \lim_{n\to\infty}\Pr_{G\sim \setG}\sqtp{\text{$G$ has the $(\alpha,\alpha)$-cover property}} &= 1.
  \end{align}
  First we verify that the conditions in \Cref{lem:conditionOfExpanders} are satisfied and then \Cref{eq:midstep21} follows.
  Clearly, $0 < \alpha < 1/\beta < 1$.
  Let $f(\Delta) = \Delta - \frac{H(\alpha) + H(\alpha \beta)}{H(\alpha) - \alpha\beta H(1/\beta)}$.
  Recall that $\Delta$ is sufficiently large.
  Thus $\alpha$ can be sufficiently small.
  Using \Cref{lem:monoOfEntropyVariant} we obtain
  \begin{align*}
    H(\alpha)-\alpha \beta H(1/\beta)
    =   H(\alpha) - 1/3H(3\alpha)
    \ge \alpha (\ln 3 - \alpha)\log_2 e
    \ge \alpha = \valueOfAlphaAtLowFugacity.
  \end{align*}
  Hence
  \begin{align*}
    f(\Delta) \ge \Delta - \frac{H(0.01)+H(1/3)}{(\ln \Delta)^2 / \Delta} \ge \Delta - \frac{\Delta}{(\ln \Delta)^2} > 0.
  \end{align*}
  Then we show that \Cref{eq:midstep22} is satisfied.
  It is equivalent to show that
  \[
    \lim_{n\to\infty}\Pr_{G\sim \setG}\sqtp{\text{$G$ does not have the $(\alpha,\alpha)$-cover property}}\to 0.
  \]
  Assume that a $\Delta$-regular bipartite graph $G=(\+L,\+R,E)$ with $n$ vertices on both sides does not have this property.
  Then there is a pair $(U,V)$ with $U\subseteq \+L,V\subseteq \+R$ or $U\subseteq \+R,V\subseteq \+L$ that $\abs{U}=\ceil{\alpha n},\abs{V}=\ceil{\alpha n}$ and $N(U)\cap V=\emptyset$.
  Applying union bound we obtain
  \begin{align}
    \label{eq:midstep18}
    &~\phantom{\le}~\Pr_{G\sim \setG}\sqtp{\text{$G$ does not have the $(\alpha,\alpha)$-cover property}} \\
    &\le 2\sum_{U\subseteq \+L\cmid \abs{U}=\ceil{\alpha n}}\sum_{V\subseteq\+R\cmid \abs{V}=\ceil{\alpha n}} \Pr_{G\sim\setG}\sqtp{N(U)\cap V = \emptyset}. \nonumber
  \end{align}
  Using \Cref{lem:boundsOfBinomCoefficients} and the perfect matching generation procedure of the distribution $\setG$, we obtain
  \begin{align*}
    \text{\Cref{eq:midstep18}}
    \le 2\binom{n}{\ceil{\alpha n}}^2 \tp{\binom{n-\ceil{\alpha n}}{\ceil{\alpha n}}\bigg/ \binom{n}{\ceil{\alpha n}} }^\Delta.
  \end{align*}
  It then follows from \Cref{lem:polyUpperBoundOfEntropy} that
  \begin{align*}
    \text{\Cref{eq:midstep18}}
    &\le 2\cdot 2^{\tp{2H(\alpha) + o(1)}n} \tp{2^{\tp{H\tp{\frac{\alpha}{1-\alpha}} + o(1)}\tp{1-\alpha}n-\tp{H(\alpha) + o(1)}n }(n+1)}^\Delta \\
    &\le 2(n+1)^\Delta\cdot
    \tp{ 2^ {2H(\alpha)+\Delta\tp{H\tp{\frac{\alpha}{1-\alpha}} \tp{1-\alpha} - H(\alpha)} +o(1) } }^n
  \end{align*}
  as $n\to\infty$.
  Recall that $\Delta$ is sufficiently large.
  Using \Cref{lem:nonPolyUpperBoundOfEntropy} and \Cref{lem:upperBoundOfAVariantOfEntropy} we obtain
  \begin{align*}
    &~\phantom{\le}~2H(\alpha)+\Delta\tp{H\tp{\frac{\alpha}{1-\alpha}} \tp{1-\alpha} - H(\alpha)}+o(1) \\
    &\le 4\alpha \log_2 \frac{1}{\alpha} - \Delta \alpha^2\log_2e + o(1) \\
    &= \frac{4(\ln \Delta)^2}{\Delta}\log_2 \frac{\Delta}{(\ln \Delta)^2} - \Delta \tp{\frac{(\ln \Delta)^2}{\Delta}}^2\log_2e + o(1) \\
    &\le \tp{\frac{4(\ln \Delta)^3}{\Delta} - \frac{(\ln \Delta)^4}{\Delta}}\log_2e + o(1)
    < C < 0
  \end{align*}
  for some constant $C=C(\Delta)<0$ as $n\to\infty$.
  Therefore
  \begin{align*}
    \text{\Cref{eq:midstep18}}
    \le 2(n+1)^\Delta 2^{C n}
    \to 0
  \end{align*}
  as $n\to\infty$.
\end{proof}

Putting together \Cref{thm:mainAtHighFugacity} and the result in this section, we obtain the following.

{\renewcommand{\thetheorem}{\ref{thm:weightedBIS}}
\begin{theorem}
  For all sufficiently large integers $\Delta$ and fugacity $\lambda=\lowerBoundOfLambdaInOmegaForm$, with high probability (tending to $1$ as $n\to \infty$) for a graph $G$ chosen uniformly at random from $\setG$, there is an FPTAS for the partition function $Z(G,\lambda)$.
\end{theorem}
\addtocounter{theorem}{-1} }

\begin{proof}
  Let $\alpha',\beta'$ be the parameters in \Cref{sec:highFugacity}.
  Let $\+G = \+G^\Delta_{\alpha',\beta'}\cap \graphClassAtLowFugacity $.
  It then follows from \Cref{lem:almostAtHighFugacity} and \Cref{lem:almostAtLowFugacity} that $\displaystyle\lim_{n\to\infty}\Pr_{G\in \setG}\sqtp{G\in \+G} = 1$.
  For $\lambda \ge 1$, we apply the algorithm from \Cref{thm:mainAtHighFugacity}.
  For $\lambda_l < \lambda < 1$, we apply the algorithm from \Cref{lem:algOfISAtLowFugacity}.
\end{proof}

Therefore, in the rest of this section, we only consider fugacity $\lambda_l < \lambda < 1$.
The notations and definitions in the rest of this section would be identical to those in \Cref{sec:highFugacity}.
So we only review needed materials briefly and state results different from those in \Cref{sec:highFugacity}.

\subsection{Approximating $Z(G, \lambda)$}

Recall that
\begin{align*}
  \clusterOfIS &= \set{I\in \+I(G)\cmid \abs{I\cap \+X}< \alpha n},
  Z_{\+X}(G,\lambda) = \sum_{I\in \clusterOfIS} \lambda^{\abs I}.
\end{align*}

The main result in this part is that we can use $Z_\+L(G,\lambda) + Z_\+R(G,\lambda)$ to approximate $Z(G, \lambda)$ for all $\lambda_l < \lambda < 1$.

\begin{lemma}\label{lem:strucApproxOfGeneralIS}
  For all sufficiently large integers $\Delta$, there are constants $C=C(\Delta)>1$ and $N=N(\Delta)$ so that for all $G\in\graphClassAtLowFugacity$ with $n > N$ vertices on both sides and $\lambda_l<\lambda<1$,
  $Z_\+L(G,\lambda)+Z_\+R(G,\lambda)$ is a $C^{-n}$-relative approximation to $Z(G, \lambda)$.
\end{lemma}

\begin{proof}
  In this proof we only consider sufficiently large integers $\Delta$.
  Applying \Cref{lem:strucApprox1AtLowFugacity}, it suffices to show that $Z_\+L(G,\lambda) + Z_\+R(G,\lambda)$ is a $C^{-n}$-relative approximation to $\sum_{I\in \+I_\+L\cup \+I_\+R} \lambda^{\abs I}$.
  For any $I\in \+I_\+L\cap \+I_\+R$, it holds that $\abs{I\cap \+L}< \alpha n$ and $\abs{I\cap \+R}< \alpha n$.
  Clearly $\sum_{I\in \+I_\+L\cup \+I_\+R}\lambda^\abs{I}\ge (\lambda+1)^n$.
  Using $\alpha \le 1/2$, \Cref{lem:boundsOfBinomCoefficients} and $\lambda_l < \lambda < 1$ we obtain
  \begin{align}
    \label{eq:midstep23}
    \frac{\sum_{I\in \+I_\+L\cap \+I_\+R}\lambda^\abs{I}}{\sum_{I\in \+I_\+L\cup \+I_\+R}\lambda^\abs{I}}
    \le (\lambda+1)^{-n}\tp{\sum_{k=0}^{\floor{\alpha n}} \binom{n}{k} \lambda^k}^2
    \le (\lambda+1)^{-n}\tp{\sum_{k=0}^{\floor{\alpha n}} \binom{n}{k}}^2
    \le n^2 \tp{\frac{4^{H(\alpha)}}{\lambda_l+1}}^n.
  \end{align}
  Recall that $\Delta$ is sufficiently large, $\alpha = \valueOfAlphaAtLowFugacity$ and $\lambda_l = \valueOfLambdaL$.
  Using \Cref{lem:nonPolyUpperBoundOfEntropy} and $\ln(x+1)\ge x/2 $ for any $0\le x \le 1$ we obtain
  \begin{align*}
    \ln \frac{4^{H(\alpha)}}{\lambda_l+1}
    = H(\alpha)\ln 4 - \ln(\lambda_l+1)
    &\le 2\alpha \log_2 \frac{1}{\alpha} \ln 4 - \lambda_l /2 \\
    &= \frac{4(\ln \Delta)^2}{\Delta}\ln \frac{\Delta}{(\ln \Delta)^2} - \frac{(\ln\Delta)^4}{2\Delta} \\
    &\le \frac{4(\ln \Delta)^3}{\Delta} - \frac{(\ln\Delta)^4}{2\Delta} \\
    &< C_1 < 0
  \end{align*}
  for some constant $C_1=C_1(\Delta)<0$.
  Therefore
  \begin{align*}
    \text{\Cref{eq:midstep23}}
    < n^2 \tp{e^{-C_1}}^{-n} = \tp{\frac{e^{-C_1}}{n^{2/n}}}^{-n} < C^{-n}
  \end{align*}
  for another constant $C=C(\Delta)>1$ and for all $n > N$ where $N=N(\Delta)$ is a sufficiently large constant.
  Using the upper bound on \Cref{eq:midstep23} and $1+x\le \exp(x)$ for any $x\in \=R$ we obtain
  \begin{align*}
    \sum_{I\in \+I_\+L\cup \+I_\+R} \lambda^{\abs I}
    \le Z_\+L + Z_\+R
    =\sum_{I\in \+I_\+L\cup \+I_\+R} \lambda^{\abs I} + \sum_{I\in \+I_\+L\cap \+I_\+R} \lambda^{\abs I}
    \le \exp(C^{-n}) \sum_{I\in \+I_\+L\cup \+I_\+R} \lambda^{\abs I}
  \end{align*}
  for all $n > N$.
\end{proof}

\begin{lemma} \label{lem:strucApprox1AtLowFugacity}
  For $\Delta\ge 3$, $G\in\graphClassAtLowFugacity$ and $\lambda\in\=R$,
  $\sum_{I\in \+I_\+L(G)\cup \+I_\+R(G)}\lambda^{\abs{I}}=Z(G,\lambda)$.
\end{lemma}

\begin{proof}
  Let $\+B = \+I \setminus \tp{\+I_\+L\cup \+I_\+R}$.
  If suffices to show that $\+B=\emptyset$.
  Suppose $\+B$ is not empty.
  Then there is an independent set $I\in \+B$ such that $\abs{I\cap \+L}\ge \alpha n$ and $\abs{I\cap \+R}\ge \alpha n$.
  Applying the cover property, we obtain that $\abs{I\cap \+R} \le \abs{\+R\setminus N(I\cap \+L)} < \alpha n$, which contradicts that $\abs{I\cap \+R}\ge \alpha n$.
  Thus $\+B=\emptyset$.
\end{proof}



\subsection{Approximating $Z_\+X(G,\lambda)$}

Recall that for all $G=(\+L,\+R,E)\in\graphClassAtLowFugacity$ with $n$ vertices on both sides and $\+X\in\set{\+L,\+R}$, we defined a polymer model $(\Gamma^*_\+X(G),w)$ of the graph $G^2$. The partition function of this model is denoted by
\begin{align*}
  \Xi(z) = \sum_{\Gamma\in \+S(\Gamma^*_\+X(G))}\prod_{\gamma\in\Gamma}w(\gamma,1)
\end{align*}
where $z$ is a complex variable and $w(\gamma,1) = \lambda^{\aabs{\gamma}}(\lambda+1)^{-\abs{N(\ol\gamma)}} z^\aabs{\gamma}$.

\begin{lemma} \label{lem:approxRepAtLowFugacity}
  For all sufficiently large integers $\Delta$,
  there are constants $C=C(\Delta)>1$ and $N=N(\Delta)$ so that for all $G=(\+L,\+R, E)\in \graphClassAtLowFugacity$ with $n>N$ vertices on both sides, $\+X\in \set{\+L,\+R}$ and $\lambda_l < \lambda < 1$,
  \begin{align*}
    (\lambda+1)^n \Xi(1) = (\lambda+1)^n \sum_{\Gamma\in\+S(\Gamma^*_\+X(G))} \prod_{\gamma\in \Gamma} w(\gamma,1)
  \end{align*}
  is a $C^{-n}$-relative approximation to $Z_\+X(G,\lambda)$.
\end{lemma}

\begin{proof}
  In this proof we only consider sufficiently large integers $\Delta$.
  It is clear that $Z_\+X(G,\lambda)\ge (\lambda+1)^n$.
  Then using \Cref{lem:exactRepOfIS} and the cover property we obtain
  \begin{align}
    \label{eq:midstep24}
    \frac{(\lambda+1)^n \Xi(1) - Z_\+X(G,\lambda)}{Z_\+X(G,\lambda)}
    &\le \sum_{\Gamma\in\+S(\Gamma^*)\cmid \aabs\Gamma \ge \alpha n}
    \prod_{\gamma\in\Gamma} w(\gamma,1)
    \le \sum_{\Gamma\in\+S(\Gamma^*)\cmid \aabs\Gamma \ge \alpha n}
    \lambda^{\aabs{\Gamma}} \tp{\lambda+1}^{(\alpha-1) n}.
  \end{align}
  For any $\gamma$, since $\abs{\gamma}< \alpha n$, it follows from the expansion property that $\abs{N_G(\ol \gamma)} \ge \beta \aabs{\gamma}$.
  The compatibility of $\Gamma$ states that $d_G(\ol{\gamma_1}, \ol{\gamma_2})>2$ for any $\gamma_1\neq \gamma_2$ in $\Gamma$, implying $N_G(\ol{\gamma_1})\cap N_G(\ol{\gamma_2}) = \emptyset$.
  Using these two facts, for any $\Gamma\in\+S(\Gamma^*)$,
  \begin{align*}
    \beta\aabs{\Gamma}
    = \beta\sum_{\gamma\in\Gamma} \abs{\ol \gamma} \le
    \sum_{\gamma\in \Gamma} \abs{N_G(\ol\gamma)}
    \le n,
  \end{align*}
  implying that $\aabs{\Gamma} \le n/\beta$.
  To enumerate each $\Gamma\in\+S(\Gamma^*)$ with $\aabs\Gamma \ge \alpha n$ at least once, we first enumerate an integer $\alpha n \le k \le n/\beta$, then since $\ol\Gamma\subseteq \+X$, we choose $k$ vertices from $\+X$.
  Recall that $\Delta$ is sufficiently large.
  Using \Cref{lem:boundsOfBinomCoefficients}, $\alpha < 1/\beta\le 1/2$, $\alpha\beta = 1/3$ and $\lambda_l < \lambda < 1$ we obtain
  \begin{align*}
    \text{\Cref{eq:midstep24}}
    \le \sum_{k=\ceil{\alpha n}}^{\floor{n/\beta}} \binom{n}{k}  \lambda^{k}\tp{\lambda+1}^{(\alpha-1) n}
    \le n \tp{\frac{2^{H(1/\beta)}}{(\lambda+1)^{1-\alpha}}}^n
    &\le n \tp{\frac{2^{H(3\alpha)}}{(\lambda_l+1)^{1-\alpha}}}^n.
  \end{align*}
  Recall that $\alpha = \valueOfAlphaAtLowFugacity$ and $\lambda_l = \valueOfLambdaL$.
  Using \Cref{lem:nonPolyUpperBoundOfEntropy} and $\ln(x+1)\ge x/2$ for any $0\le x\le 1$ we obtain
  \begin{align*}
    \ln \frac{2^{H(3\alpha)}}{(\lambda_l + 1)^{1-\alpha}}
    = H(3\alpha) \ln 2 - (1-\alpha)\ln(\lambda_l+1)
    &\le 6\alpha \log_2 \frac{1}{3\alpha}\ln 2 - \lambda_l / 4 \\
    &= \frac{6(\ln \Delta)^2}{\Delta}\ln \frac{\Delta}{3(\ln\Delta)^2} - \frac{(\ln\Delta)^4}{4\Delta} \\
    &\le \frac{6(\ln\Delta)^3}{\Delta} - \frac{(\ln\Delta)^4}{4\Delta} \\
    &< C_1 < 0
  \end{align*}
  for some constant $C_1=C_1(\Delta) < 0$.
  Therefore
  \begin{align*}
    \text{\Cref{eq:midstep24}}
    < n\tp{e^{-C_1}}^{-n} = \tp{\frac{e^{-C_1}}{n^{2/n}}}^{-n} < C^{-n}
  \end{align*}
  for some constant $C=C(\Delta) >1$ and for all $n > N$ where $N=N(\Delta)$ is a sufficiently large constant.
  Using the upper bound on \Cref{eq:midstep24} and $1+x\le \exp(x)$ for any $x\in\=R$ we obtain
  \begin{align*}
    Z_\+X(G,\lambda) &\le (\lambda+1)^n \Xi(1)
    = Z_\+X(G,\lambda) + ((\lambda+1)^n \Xi(1) - Z_\+X(G,\lambda))
    \le \exp(C^{-n}) Z_\+X(G,\lambda)
  \end{align*}
  for all $n > N$.
\end{proof}

\begin{lemma}\label{lem:decayRateAtLowFugacity}
  For all polymers $\gamma\in \Gamma^*_{\+X}(G)$ defined by $G=(\+L,\+R,E)\in \graphClassAtLowFugacity, \+X\in \set{\+L,\+R}$ and $\lambda_l < \lambda < 1$,
  $ w(\gamma,1) \le (\lambda+1)^{-\beta\aabs\gamma}$.
\end{lemma}

\begin{proof}
  For every $\gamma\in \Gamma^*$, it follows from the definition of polymers that $\aabs\gamma < \alpha n$.
  Using the expansion property we obtain
  \[
    w(\gamma,1)= \lambda^{\aabs\gamma} (\lambda+1)^{-N(\ol\gamma)}
    \le (\lambda+1)^{-\beta\aabs\gamma}. \qedhere
  \]
\end{proof}

\subsection{Approximating the partition function of the polymer model}

\begin{lemma} \label{lem:algOfLowFugacityPolymerPF}
  For all sufficiently large integers $\Delta$ and $\lambda_l < \lambda < 1$, there is an FPTAS for $\Xi(1)$ for all $G=(\+L,\+R,E)\in\graphClassAtLowFugacity$ and $\+X\in\set{\+L,\+R}$.
\end{lemma}

\begin{proof}
  We use the FPTAS in \Cref{thm:alg} to design the FPTAS we need.
  To this end, we generate a graph $G^2$ in polynomial time in $\abs{G}$ for any $G\in\graphClassAtLowFugacity$.
  We use this new graph $G^2$ as input to the FPTAS in \Cref{thm:alg}.
  It is straightforward to verify the first three conditions in \Cref{thm:alg}, only with the exception that the information of $G^2$ may not be enough because certain connectivity information in $G$ is discarded in $G^2$.
  Nevertheless, we can use the original graph $G$ whenever needed and thus the first three conditions are satisfied.
  For the last condition, \Cref{lem:zeroFreeOfLowFugacityPolymerPF} verifies it.
\end{proof}

\begin{lemma} \label{lem:zeroFreeOfLowFugacityPolymerPF}
  There is a constant $R>1$ so that for all sufficiently large integers $\Delta,G=(\+L,\+R,E)$$\in\graphClassAtLowFugacity,\+X\in \set{\+L,\+R}$ and $z\in\=C$ with $\abs{z} <R$, $\Xi(z)\neq 0$.
\end{lemma}

\begin{proof}
  In this proof we only consider sufficiently large integers $\Delta$.
  Set $R=2$.
  For any $\gamma$, let $a(\gamma)=\aabs{\gamma}$.
  We will verify that the KP-condition
  \begin{align}
    \sum_{\gamma:\gamma\not\sim\gamma^*} e^{\aabs\gamma} \abs{w(\gamma,z)} \le \aabs{\gamma^*} \label{eq:midstep16}
  \end{align}
  holds for any $\gamma^*$ and any $\abs{z} < R$.
  It then follows from \Cref{lem:KP-condition} that $\Xi(z)\neq 0$ for any $\abs{z} < R$.
  Recall that $d_{G^2}(\ol\gamma, \ol{\gamma^*})\le 1$ for all $\gamma\not\sim\gamma^*$.
  Thus there is always a vertex $v\in \ol\gamma\subseteq \+X$ such that $v\in \ol{\gamma^*}\sqcup N_{G^2}(\ol{\gamma^*})$.
  The number of such vertices $v$ is at most $\Delta^2\aabs{\gamma^*}$.
  So to enumerate each $\gamma\not\sim\gamma^*$ at least once, we can
  \begin{enumerate}[label=\alph*)]
    \item first enumerate a vertex $v$ in $\+X\cap \tp{\ol{\gamma^*}\cup N_{G^2}(\ol{\gamma^*})}$;
    \item then enumerate an integer $k$ from $1$ to $\floor{\alpha n}$;
    \item finally enumerate $\gamma$ with $v\in \ol\gamma$ and $\aabs{\gamma}=k$.
  \end{enumerate}
  Since $\ol\gamma$ is connected in $G^2$, using \Cref{lem:EnumSubgraph} and \Cref{lem:decayRateAtLowFugacity} and $\lambda_l < \lambda < 1$ we obtain
  \begin{align}
    \label{eq:midstep17}
    \sum_{\gamma:\gamma\not\sim\gamma^*}e^{\aabs\gamma} \abs{w(\gamma,z)}
    \le \sum_{\gamma:\gamma\not\sim\gamma^*}e^{\aabs\gamma} \abs{w(\gamma,1)}\cdot\abs{z}^{\aabs \gamma}
    &\le \Delta^2 \aabs{\gamma^*} \sum_{k=1}^{\floor{\alpha n}}(e\Delta^2)^{k-1}e^k (\lambda+1)^{-\beta k}R^k \\
    &\le \aabs{\gamma^*}\sum_{k=1}^\infty \tp{e^2\Delta^2(\lambda_l+1)^{-\beta}R}^k. \nonumber
  \end{align}
  Recall that $\Delta$ is sufficiently large, $\beta = \valueOfBetaAtLowFugacity = \frac{\Delta}{3(\ln \Delta)^2}$ and $\lambda_l = \valueOfLambdaL$.
  Using $\ln (x+1)\ge x/2$ for any $0\le x\le 1$ we obtain
  \begin{align*}
    \ln \tp{e^2\Delta^2 (\lambda_l + 1)^{-\beta}R}
    &= 2 + 2\ln \Delta - \beta \ln (\lambda_l + 1) + \ln R \\
    &\le 2 + 2 \ln \Delta - \frac{\Delta}{3(\ln \Delta)^2}\cdot \frac{(\ln\Delta)^4}{2\Delta} + \ln R \\
    &= 2 \ln \Delta - \frac{(\ln \Delta)^2}{6} + 2 + \ln 2 \\
    &< -1.
  \end{align*}
  Therefore
  \begin{align*}
    \text{\Cref{eq:midstep17}} \le \aabs{\gamma^*} \sum_{k=1}^\infty e^{-k} = \frac{1}{e-1} \aabs{\gamma^*} < \aabs{\gamma^*},
  \end{align*}
  which proves \Cref{eq:midstep16}.
\end{proof}




\begin{lemma} \label{lem:algOfISAtLowFugacity}
  For all sufficiently large integers $\Delta$ and $\lambda_l < \lambda < 1$, there is an FPTAS for $Z(G, \lambda)$ for all $G\in\graphClassAtLowFugacity$.
\end{lemma}

\begin{proof}
  This can be readily obtained by replacing facts used in the proof of \Cref{lem:algOfIS} with corresponding results obtained in this section.
\end{proof}

\section{Counting colorings}

Throughout this section, we consider integers $q\ge 3, \Delta\ge \lowerBoundOfDeltaForColorings$ and set parameters $\sOfColorings, \alpha, \beta$ to be
\begin{align*}
  \sOfColorings = \valueOfS, \alpha = \valueOfAlphaForColorings, \beta = \valueOfBetaForColorings.
\end{align*}
We define a set $\graphClassOfColorings$ of graphs as
\begin{align*}
  \graphClassOfColorings = \set{G\in\graphClassAtHighFugacity\cmid \text{$G$ has the $(s,\alpha / q)$-cover property}}.
\end{align*}

\begin{lemma}\label{lem:almostOfColorings}
  For $q\ge 3$ and $\Delta\ge \lowerBoundOfDeltaForColorings$, $\displaystyle\lim_{n\to\infty}\Pr_{G\sim \setG}\sqtp{G\in \graphClassOfColorings}=1$.
\end{lemma}

\begin{proof}
  Recall that $s = \valueOfS, \alpha = \valueOfAlphaForColorings$ and $\beta = \valueOfBetaForColorings$.
  It suffices to show that
  \begin{align}
    \label{eq:midstep25}
    &\lim_{n\to\infty} \Pr_{G\sim \setG}\sqtp{G\in \graphClassAtHighFugacity} = 1, \\
    \label{eq:midstep26}
    &\lim_{n\to\infty} \Pr_{G\sim \setG}\sqtp{\text{$G$ has the $(s,\alpha/q)$-cover property}} = 1.
  \end{align}
  First we verify that the conditions in \Cref{lem:conditionOfExpanders} are satisfied and then \Cref{eq:midstep25} follows.
  Let $f(\Delta) = \Delta - \frac{H(\alpha) + H(\alpha \beta)}{H(\alpha) - \alpha\beta H(1/\beta)}$.
  It follows from \Cref{lem:monoOfEntropyVariant} that
  \begin{align*}
    H(1/\Delta^{1/2}) - 1/3 H(3/\Delta^{1/2}) \ge 1/\Delta^{1/2} \tp{\ln 3 - 1/\Delta^{1/2}}\log_2 e > 1.2/\Delta^{1/2}
  \end{align*}
  for any $\Delta \ge 4$.
  Then
  \begin{align*}
    f(\Delta) \ge \Delta - \frac{H(1/100) + H(1/3)}{H(1/\Delta) - 1/3H(3/\Delta)} \ge 0.1\Delta > 0
  \end{align*}
  for any $\Delta \ge 100$.
  Then we show that \Cref{eq:midstep26} is satisfied. It is equivalent to show that
  \begin{align*}
    \lim_{n\to\infty} \Pr_{G\sim \setG}\sqtp{\text{$G$ does not have the $(s,\alpha/q)$-cover property}} = 0.
  \end{align*}
  Assume that a $\Delta$-regular bipartite graph $G=(\+L,\+R,E)$ with $n$ vertices on both sides does not have the $(s,\alpha / q)$-cover property.
  Then there is a pair $(U,V)$ with $U\subseteq \+L,V\subseteq \+R$ or $U\subseteq \+R,V\subseteq \+L$ that $\abs{U}=\ceil{sn},\abs{V}=\ceil{\alpha / q n}$ and $N(U)\cap V=\emptyset$.
  Thus
  \begin{align}
    \label{eq:midstep15}
    &~\phantom{\le}~\Pr_{G\sim \setG}\sqtp{\text{$G$ does not have the $(s,\alpha / q)$-cover property}} \\
    &\le 2\sum_{U\subseteq \+L\cmid \abs{U}=\ceil{sn}}\sum_{V\subseteq\+R\cmid \abs{V}=\ceil{\alpha / qn}} \Pr_{G\sim\setG}\sqtp{N(U)\cap V = \emptyset}. \nonumber
  \end{align}
  Using \Cref{lem:boundsOfBinomCoefficients} and the perfect matching generation procedure of the distribution $\setG$, we obtain
  \begin{align*}
    \text{\Cref{eq:midstep15}}
    \le 2\binom{n}{\ceil{sn}}\binom{n}{\ceil{\alpha / q n}} \tp{\binom{n-\ceil{\alpha / qn}}{\ceil{sn}}\bigg/ \binom{n}{\ceil{sn}} }^\Delta.
  \end{align*}
  Recall that $s=\valueOfS$ and $\alpha = \valueOfAlphaForColorings$.
  It then follows from \Cref{lem:polyUpperBoundOfEntropy} that
  \begin{align*}
    \text{\Cref{eq:midstep15}}
    &\le 2\cdot 2^{\tp{H(s)+H(\frac{\alpha}{q}) + o(1)}n} \tp{2^{\tp{H\tp{\frac{s}{1-\alpha / q}} + o(1)}\tp{1-\frac{\alpha}{q} + o(1)}n-\tp{H(s) + o(1)}n }(n+1)}^\Delta \\
    &\le 2(n+1)^\Delta\cdot
    \tp{ 2^ {1+\Delta\tp{H\tp{\frac{s}{1-\alpha / q}} \tp{1-\frac{\alpha}{q}} - H(s)+o(1)} } }^n
  \end{align*}
  for all sufficiently large $n$.
  Using \Cref{lem:upperBoundOfAVariantOfEntropy} we obtain
  \begin{align*}
    1+\Delta\tp{H\tp{\frac{s}{1-\alpha / q}} \tp{1-\frac{\alpha}{q}} - H(s)+o(1)}
    &\le
    1 - \Delta \tp{s\alpha / q \log_2 e + o(1)}\\
    &\le 1 - \lowerBoundOfDeltaForColorings \valueOfS / q \log_2e /2\\
    &\le 1 - \frac{25\ol q^3}{9}\log_2 e\\
    & < 1/C
  \end{align*}
  for some constant $C>1$ and for all sufficiently large $n$.
  Therefore
  \begin{align*}
    \text{\Cref{eq:midstep15}}
    \le 2(n+1)^\Delta C^{-n}
    \to 0
  \end{align*}
  as $n\to\infty$.
\end{proof}

In the rest of this section, whenever possible, we will simplify notations by omitting superscripts, subscripts and brackets with the symbols between (but this will not happen in the statement of lemmas and theorems).
For example, $\+C(G)$ may be written as $\+C$ if $G$ is clear from context.

\subsection{Approximating $\abs{\+C(G)}$}

For all $q\ge 3, \Delta\ge 3, G=(\+L,\+R,E)\in \graphClassOfColorings$ and $\emptyset\subsetneq X\subsetneq [q]$, we define
\begin{align*}
  \clusterOfColoring = \set{\sigma\in \+C(G)\cmid d_X(\sigma) < \alpha n}
\end{align*}
where $d_X(\sigma) = \abs{\sigma_\+L^{-1}([q]\setminus X)} + \abs{\sigma_\+R^{-1}(X)}$ (recall that $\sigma_\+L = \sigma|_\+L$ and $\sigma_\+R=\sigma|_\+R$).
The main result of this subsection is that we can use $\sum_{X\cmid \abs{X}\in \set{\ul q, \ol q}}\abs{\clusterOfColoring}$ to approximate $\abs{\+C(G)}$.

\begin{lemma} \label{lem:strucApproxOfColorings}
  For $q\ge 3$ and $\Delta\ge\lowerBoundOfDeltaForColorings$, there are constants $C=C(q) > 1$ and $N=N(q)$ such that for all $G\in \graphClassOfColorings$ with $n>N$ vertices on both sides,
  $Z$ is a $C^{-n}$-relative approximation to $\abs{\+C(G)}$, where $Z = \binom{q}{\ul{q}}\abs{\clusterOfColoring[{[\ul q]}]}$ if $q$ is even, otherwise $Z = \binom{q}{\ul q}\tp{\abs{\clusterOfColoring[{[\ul q]}]} + \abs{\clusterOfColoring[{[\ol q]}]}}$.
\end{lemma}

\begin{proof}
  Let $N_1,C_1,N_2,C_2$ and $N_3,C_3$ be the constants in \Cref{lem:strucApprox1OfColorings}, \Cref{lem:strucApprox2OfColorings} and \Cref{lem:strucApprox3OfColorings}, respectively.
  It follows from these lemmas that
  \begin{align*}
    \exp(-(C_1^{-n}+C_2^{-n}+C_3^{-n}))Z \le \abs{\+C} \le \exp(C_1^{-n}+C_2^{-n}+C_3^{-n})Z
  \end{align*}
  for all $n > \max(N_1,N_2,N_3)$.
  It is clear that
  \[C_1^{-n} + C_2^{-n} + C_3^{-n} \le 3\min(C_1,C_2,C_3)^{-n} = \tp{\frac{\min(C_1,C_2,C_3)}{3^{1/n}}}^{-n} < C^{-n}\]
  for another constant $C=C(q)>1$ and for all $n > N \ge \max(N_1,N_2,N_3)$ where $N=N(q)$ is another sufficiently large constant.
  Therefore we obtain
  \begin{align*}
    \exp(-C^{-n}) Z \le \abs{\+C} \le \exp(C^{-n})Z
  \end{align*}
  for all $n >N$.
\end{proof}

\begin{lemma} \label{lem:strucApprox1OfColorings}
  For $q\ge 3$ and $\Delta\ge\lowerBoundOfDeltaForColorings$, there are constants $C=C(q) > 1$ and $N=N(q)$ such that for all $G\in \graphClassOfColorings$ with $n>N$ vertices on both sides,
  $\abs{\bigcup_{X\cmid\emptyset\subsetneq X\subsetneq [q]} \clusterOfColoring}$ is a $C^{-n}$-relative approximation to $\abs{\+C(G)}$.
\end{lemma}

\begin{proof}
  For any coloring $\omega$, let
  \begin{align*}
    \-{maj}(\omega) = \set{c\in [q] \cmid \abs{\omega^{-1}(c)}\ge sn}.
  \end{align*}
  Fix $\sigma\in \+C$.
  If $\-{maj}(\sigma_\+L)\cap \-{maj}(\sigma_\+R) \neq \emptyset$, then there exists a color $c\in [q]$ that $\abs{\sigma_\+L^{-1}(c)} \ge sn$ and $\abs{\sigma_\+R^{-1}(c)} \ge sn$.
  Since $\abs{\sigma^{-1}_\+L(c)}\ge sn$, it follows from the cover property that $\abs{N(\sigma^{-1}_\+L(c))}> (1-\alpha / q)n$.
  Since $\sigma$ is proper, then $\abs{\sigma_\+R^{-1}(c)} \le n - \abs{N(\sigma^{-1}_\+L(c))} <\alpha / q n < sn$, which contradicts that $\abs{\sigma_\+R^{-1}(c)} \ge sn$.
  Therefore, $\-{maj}(\sigma_\+L)\cap \-{maj}(\sigma_\+R) =\emptyset$ for any $\sigma\in \+C$.
  Let $\+B = \set{\sigma \in \+C \cmid \sigma\not\in \cup_X\+C_X}$.
  We claim that $\abs{\-{maj}(\sigma_\+{L})} + \abs{\-{maj}(\sigma_\+{R})} \le q - 1$ for any $\sigma\in \+B$. Suppose that $\abs{\-{maj}(\sigma_\+{L})} + \abs{\-{maj}(\sigma_\+{R})} = q$ for some $\sigma$. Let $X= \-{maj}(\sigma_\+{L})$. Then we have
  \begin{align*}
    d_X(\sigma)
    = \abs{\sigma_\+L^{-1}([q]\setminus X)} + \abs{\sigma_\+R^{-1}(X)}
    &= \sum_{c\in \-{maj}(\sigma_\+R)}\abs{\sigma_\+L^{-1}(c)} + \sum_{c\in \-{maj}(\sigma_\+L)} \abs{\sigma_\+R^{-1}(c)} \\
    &\le \sum_{c\in \-{maj}(\sigma_\+R)}\tp{n - \abs{N(\sigma_\+R^{-1}(c))}} + \sum_{c\in \-{maj}(\sigma_\+L)} \tp{n-\abs{N(\sigma_\+L^{-1}(c))}} \\
    &< \alpha n.
  \end{align*}
  By definition $\sigma\in C_{X}(G)$ and thus $\sigma\notin \+B$.

  We give an upper bound of $\abs{\+B}$ via the following procedure which enumerates each $\sigma \in \+B$ at least once.
  \begin{enumerate}[label=\alph*)]
    \item Recall that $\abs{\-{maj}(\sigma_\+{L})\sqcup \-{maj}(\sigma_\+{R})}\le q-1$ for any $\sigma\in \+B$.
    Thus we enumerate two sets $A,B\subseteq [q]$ such that $\abs{A\sqcup B}= q-1$.
    Clearly, there are at most $q 2^q$ ways to enumerate such sets.

    \item Assume that $A$ and $B$ have been enumerated out.
    Then we enumerate colorings $\sigma\in \+B$ with $\-{maj}(\sigma_\+{L})\subseteq A$ and $\-{maj}(\sigma_\+{R})\subseteq B$.
    To this end, we can enumerate $\sigma_\+L$ and $\sigma_\+R$ independently and combine them together.

    \item Consider $\sigma_\+L$ with $\-{maj}(\sigma_\+L)\subseteq A$.
    Clearly, $\abs{\sigma_\+L^{-1}([q]\setminus A)}\le (q-\abs{A})sn$.
    Thus we enumerate a set $\+L_\-{minor}\subseteq \+L$ with size $\floor{(q-\abs{A})sn}$.
    Since $qs\le 1/2$, there are at most $\binom{n}{\floor{(q-\abs{A})sn}} \le \binom{n}{\floor{qsn}}$ ways to enumerate such a set.

    \item Assume that $\+L_\-{minor}$ has been enumerated out.
    Then we count colorings $\sigma\in \+B$ with $\sigma_\+L^{-1}([q]\setminus A)\subseteq \+L_\-{minor}$.
    The number of such colorings is upper bounded by $q^{(q-\abs{A})sn}\abs{A}^n$.

    \item Putting c) and d) together, there are at most $\binom{n}{\floor{qsn}} q^{(q-\abs{A})sn} \abs{A}^n$ ways to enumerate colorings $\sigma_\+{L}$ with $\-{maj}(\sigma_\+L)\subseteq A$.
    Analogously, there are at most $\binom{n}{\floor{qsn}}q^{(q-\abs{B})sn} \abs{B}^n$ ways to enumerate colorings $\sigma_\+R$ with $\-{maj}(\sigma_\+R)\subseteq B$.

    \item Combining all the previous steps, we obtain that
    \begin{align*}
      \abs{\+B} \le q2^q \binom{n}{\floor{qsn}}^2 q^{(2q-\abs{A}-\abs{B})sn} \abs{A}^n \abs{B}^n
      \le q2^q 4^{H(q s)n} q^{(q+1)sn}\ul q^n (\ol q - 1)^n,
    \end{align*}
    where the inequality follows from \Cref{lem:boundsOfBinomCoefficients}.
  \end{enumerate}
  Clearly $\abs{\cup_X \+C_X}\ge \ul q^n \ol q^n$ and we obtain
  \begin{align}
    \label{eq:midstep10}
    \frac{\abs{\+B}}{\abs{\cup_X \+C_X}}
    &\le q2^q\tp{4^{H(qs)}q^{(q+1)s} (1-1/\ol q)}^n.
  \end{align}
  Recall that $s = \valueOfS$.
  It holds that $q s \le \frac{1}{9\ol q^4}$.
  Using \Cref{lem:polyUpperBoundOfEntropy}, $\ln(1+x)\le x$ for any $x>-1$ and $\ol q \ge 2$ we obtain
  \begin{align}
    \label{eq:midstep11}
    4^{H(q s)}q^{(q+1) s}(1-1/\ol q)
    &\le 16^\frac{1}{3\ol q^2}q^{\frac{1}{9\ol q^4}+\valueOfS}(1-1/\ol q) \\
    & \le \exp\tp{\frac{\ln 16}{3\ol q^2} + \frac{\ln q}{9 \ol q^4} + \frac{\ln q}{18\ol q^5}- \frac{1}{\ol q}} \nonumber\\
    & \le \exp\tp{\tp{\frac{\ln 16}{3\times 2} + \frac{1}{9\times 4} + \frac{1}{18\times 8}- 1}\frac{1}{\ol q}} \nonumber \\
    & < \exp\tp{-\frac{1}{2\ol q}} \nonumber\\
    & < 1/C_1 \nonumber
  \end{align}
  for some constant $C_1=C_1(q) > 1$.
  Therefore,
  \begin{align*}
    \text{\Cref{eq:midstep10}}
    \le q2^q C_1^{-n} = \tp{\frac{C_1}{\tp{q2^q}^{1/n}}}^{-n} < C^{-n}
  \end{align*}
  for another constant $C=C(q)>1$ and $n>N$ where $N=N(q)$ is a sufficiently large constant.
  Using the upper bound on \Cref{eq:midstep10} and $1+x\le \exp(x)$ for any $x\in \=R$ we obtain
  \begin{align*}
    \abs{\cup_X\+C_X}
    \le \abs{\+C}
    = \abs{\cup_X\+C_X} + \abs{\+B}
    \le \exp(C^{-n})\abs{\cup_X\+C_X}
  \end{align*}
  for all $n > N$.
\end{proof}

\begin{lemma} \label{lem:strucApprox2OfColorings}
  For $q\ge 3$ and $\Delta\ge\lowerBoundOfDeltaForColorings$, there are constants $C=C(q) > 1$ and $N=N(q)$ such that for all $G\in \graphClassOfColorings$ with $n>N$ vertices on both sides,
  $\sum_{X\cmid\emptyset\subsetneq X\subsetneq [q]}\abs{\clusterOfColoring}$ is a $C^{-n}$-relative approximation to $\abs{\bigcup_{X\cmid\emptyset\subsetneq X\subsetneq [q]} \clusterOfColoring}$.
\end{lemma}

\begin{proof}
  Fix two sets $\emptyset \subsetneq X\neq Y\subsetneq [q]$.
  Clearly, $\abs{X\cap Y} + \abs{[q]\setminus (X\cup Y)}\le (\max(\abs X,\abs Y)-1) +(q-\max(\abs X,\abs Y))= q -1$.
  For any $\sigma\in \+C_X\cap \+C_Y$, it holds that
  \begin{align*}
    \abs{\sigma_\+L^{-1}([q]\setminus (X\cap Y))} + \abs{\sigma_\+R^{-1}(X\cup Y)}
    &\le \tp{\abs{\sigma_\+L^{-1}([q]\setminus X)} + \abs{\sigma_\+L^{-1}([q]\setminus Y)}} + \tp{\abs{\sigma_\+R^{-1}(X)} + \abs{\sigma_\+R^{-1}(Y)}}\\
    &= \tp{\abs{\sigma_\+L^{-1}([q]\setminus X)} + \abs{\sigma_\+R^{-1}(X)}} + \tp{\abs{\sigma_\+L^{-1}([q]\setminus Y)} + \abs{\sigma_\+R^{-1}(Y)}} \\
    &< 2\alpha n.
  \end{align*}
  This shows that for $\sigma\in \+C_X\cap\+C_Y$ most of the vertices in $\+L$ are colored using colors from $X\cap Y$ and most of the vertices in $\+R$ are colored using colors from $[q]\setminus (X\cup Y)$.
  According to this,
  we can upper bound $\abs{\+C_X\cap \+C_Y}$ via the following procedure which enumerates each $\sigma\in\+C_X\cap\+C_Y$ at least once.
  First we enumerate a set $\+B\subseteq \+L\cup \+R$ with $\abs{\+B}=\floor{2\alpha n}$.
  Then the vertices in $\+B$ can be colored arbitrarily, but the vertices in $\+L\setminus \+B$ can only be colored with colors from $X\cap Y$ and the vertices in $\+R\setminus \+B$ can only be colored with colors from $[q]\setminus (X\cup Y)$.
  Thus we obtain
  \begin{align*}
    \abs{\+C_X\cap \+C_Y}
    \le
    \binom{2n}{\floor{2\alpha n}}q^{2\alpha n}\abs{X\cap Y}^n \abs{[q]\setminus(X\cup Y)}^n
    &\le \tp{4^{H(\alpha)} q^{2\alpha} \ul q (\ol q - 1)}^n,
  \end{align*}
  where the inequality follows from \Cref{lem:boundsOfBinomCoefficients} and $\abs{X\cap Y} + \abs{[q]\setminus (X\cup Y)}\le q-1$.
  It is clear that $\abs{\cup_X \+C_X}\ge \ul q^n \ol q^n$ and we obtain
  \begin{align*}
    \frac{\abs{\+C_X\cap \+C_Y}}{\abs{\cup_X\+C_X}}
    &\le \tp{4^{H(\alpha)}q^{2\alpha}(1-1/\ol q)}^n.
  \end{align*}
  Recall that $s=\valueOfS$ and $\alpha = \valueOfAlphaForColorings \le \upperBoundOfAlphaForColorings$.
  Since $\alpha\le q s \le 1/2$ and $2\alpha \le (q+1)s$, it follows from the upper bound on \Cref{eq:midstep11} that
  \begin{align*}
    4^{H(\alpha)}q^{2\alpha}(1-1/\ol q) \le 4^{H(q s)}q^{(q+1)s}(1-1/\ol q) < 1/C_1
  \end{align*}
  for some constant $C_1=C_1(q) > 1$.
  Therefore
  \begin{align}
    \label{eq:midstep14}\frac{\sum_{X\neq Y} \abs{\+C_X\cap \+C_Y}}{\abs{\cup_X\+C_X}}
    \le 4^q C_1^{-n} \le \tp{\frac{C_1}{4^{q/n}}}^{-n} <C^{-n}
  \end{align}
  for another constant $C=C(q)>1$ and $n > N$ where $N=N(q)$ is a sufficiently large constant.
  Using the upper bound on \Cref{eq:midstep14} and $1+x\le \exp(x)$ for any $x\in \=R$ we obtain
  \begin{align*}
    \abs{\cup_X \+C_X}
    \le \sum_X\abs{\+C_X}
    \le \abs{\cup_X \+C_X} + \sum_{X\neq Y} \abs{\+C_X\cap \+C_Y}
    &\le \exp(C^{-n}) \abs{\cup_X\+C_X}
  \end{align*}
  for all $n>N$.
\end{proof}

\begin{lemma} \label{lem:strucApprox3OfColorings}
  For $q\ge 3$ and $\Delta\ge\lowerBoundOfDeltaForColorings$, there are constants $C=C(q) > 1$ and $N=N(q)$ such that for all $G\in \graphClassOfColorings$ with $n>N$ vertices on both sides,
  $Z$ is a $C^{-n}$-relative approximation to $\sum_{X\cmid\emptyset\subsetneq X\subsetneq [q]}$ $\abs{\clusterOfColoring}$, where $Z = \binom{q}{\ul q}\abs{\clusterOfColoring[{[\ul q]}]}$ if $q$ is even, otherwise $Z = \binom{q}{\ul q}\tp{\abs{\clusterOfColoring[{[\ul q]}]} + \abs{\clusterOfColoring[{[\ol q]}]}}$.
\end{lemma}

\begin{proof}
  It follows from the symmetry of colors that $\abs{\+C_X} = \abs{\+C_Y}$ for any $X$ and $Y$ with $\abs{X}=\abs{Y}$.
  Fix $Y$ with $\abs{Y} < \ul q$ or $\abs{Y} > \ol q$.
  We upper bound $\abs{\+C_Y}$ via the following procedure which enumerates each coloring $\sigma\in \+C_Y$ at least once.
  For each $\sigma\in \+C_Y$, it holds that $d_Y(\sigma) < \alpha n$.
  Thus we can enumerate a set $\+B\subseteq\+L\cup\+R$ with $\abs{\+B} = \floor{\alpha n}$.
  The vertices in $\+B$ can be colored arbitrarily, but the colors of the vertices in $\+L\setminus \+B$ can only be chosen from $Y$ and the vertices in $\+R\setminus \+B$ can only be colored with colors from $[q]\setminus Y$.
  Thus we obtain
  \begin{align*}
    \abs{\+C_Y}
    \le
    \binom{2n}{\floor{\alpha n}}q^{\alpha n} \abs{Y}^n \abs{[q]\setminus Y}^n \le \tp{4^{H(\alpha / 2)}q^{\alpha} (\ul q- 1)(\ol q + 1)}^n,
  \end{align*}
  where the inequality follows from \Cref{lem:boundsOfBinomCoefficients} and $\abs{Y} \cdot\abs{[q]\setminus Y} \le (\ul q-1) (\ol q + 1)$.
  Clearly $Z\ge \ul q^n \ol q^n$ and we obtain
  \begin{align*}
    \frac{\abs{\+C_Y}}{Z}
    \le \tp{4^{H(\alpha / 2)}q^{\alpha} (1-1/\ul q)(1+1/\ol q)}^n
    \le \tp{4^{H(\alpha / 2)}q^{\alpha} (1-1/\ol q^2)}^n.
  \end{align*}
  Recall that $\alpha = \valueOfAlphaForColorings \le \upperBoundOfAlphaForColorings$.
  Using \Cref{lem:polyUpperBoundOfEntropy}, $\ln(1+x)\le x$ for any $x>-1$ and $\ol q\ge 2$ we obtain
  \begin{align*}
    4^{H(\alpha / 2)} q^{\alpha} (1 - 1/ \ol q^2)
    \le 16^{\frac{1}{\ol q^2\sqrt{20\ol q}}} q^{\frac{1}{10\ol q^5}} (1-1/\ol q^2)
    &\le \exp\tp{\frac{\ln 16}{\ol q^2\sqrt{20\ol q}} + \frac{\ln q}{10\ol q^5} - \frac{1}{\ol q^2}} \\
    &\le \exp\tp{\tp{\frac{\ln 16}{\sqrt{20 \times 2}} + \frac{1}{10\times 4} - 1}\frac{1}{\ol q^2}} \\
    &\le \exp\tp{-\frac{1}{2\ol q^2}}\\
    &< 1/C_1
  \end{align*}
  for some constant $C_1=C_1(q) > 1$.
  Therefore
  \begin{align}
    \label{eq:midstep12}
    \frac{\sum_{Y\cmid \abs{Y} < \ul q \lor \abs{Y} > \ol q}\abs{\+C_Y}}{Z}
    \le 2^q C_1^{-n} &\le \tp{\frac{C_1}{2^{q / n}}}^{-n} < C^{-n}
  \end{align}
  for another constant $C=C(q)>1$ and $n>N$ where $N=N(q)$ is a sufficiently large constant.
  Using the upper bound on \Cref{eq:midstep12} and $1+x\le \exp(x)$ for any $x\in\=R$ we obtain
  \begin{align*}
    Z \le \sum_X\abs{\+C_X}
    = Z + \sum_{Y\cmid \abs{Y} < \ul q \lor \abs Y > \ol q} \abs{\+C_Y} &\le \exp(C^{-n}) Z
  \end{align*}
  for all $n > N$.
\end{proof}

\subsection{Approximating $\abs{\clusterOfColoring}$} \label{sec:polymerOfColorings}

In this subsection, we discuss how to approximate $\abs{\clusterOfColoring}$ for $G=(\+L,\+R,E)\in \graphClassOfColorings$ and $X\subseteq [q]$ with $\abs{X}\in \set{\ul q, \ol q}$.
We will use the polymer model (see \Cref{sec:polymer-model}).
First we constructively define the polymers we need.
For any $\sigma\in \clusterOfColoring$, let $U = \set{v\in \+L \cmid \sigma(v) \not\in X} \cup \set{v\in \+R \cmid \sigma(v)\not\in [q]\setminus X}$.
We can partition the graph $(G^2)[U]$ into connected components $U_1,U_2,\ldots,U_k$ for some $k\ge 0$. There are no edges in $G^2$ between $U_i$ and $U_j$ for any $1\le i \neq j \le k$.
If $k > 0$, let $p(\sigma) = \set{\tp{U_1, \sigma|_{U_1}}, \tp{U_2, \sigma|_{U_2}}, \ldots, \tp{U_k, \sigma|_{U_k}}}$.
If $k = 0$, let $p(\sigma)=\emptyset$.
We define the set of all polymers to be
\begin{align*}
  \Gamma^*_X(G) = \bigcup_{\sigma\in \clusterOfColoring} p(\sigma),
\end{align*}
and each element in this set is called a polymer.
When the graph $G$ and $X$ are clear from the context, we simply denote by $\Gamma^*$ the set of polymers.
For each polymer $\gamma\in \Gamma^*$, define its weight function $w(\gamma,\cdot)$ as
\begin{align*}
  w(\gamma, z) = \frac{\abs{\clusterOfColoring[\gamma]}}{\abs{X}^n \tp{q - \abs X}^n}z^{\aabs{\gamma}},
\end{align*}
where $z$ is a complex variable and
\begin{align*}
  \clusterOfColoring[\gamma] =
  \set{\sigma\in \+C_X(G)\cmid
  \sigma|_{\ol \gamma} = \conf \land
  \sigma(\+L\setminus \ol\gamma) \subseteq X \land
  \sigma(\+R\setminus \ol\gamma) \subseteq [q]\setminus X}.
\end{align*}
The number of colorings in $\clusterOfColoring[\gamma]$ can be computed in polynomial time in $\aabs{\gamma}$ since $\abs{N(\ol \gamma)}\le \beta\aabs\gamma$ and
\begin{align*}
  \abs{\+C_\gamma(G)} = \tp{\prod_{v\in \+L} \abs{X \setminus \omega_{\ol \gamma}(N(v)\cap V(\ol \gamma))}}\tp{\prod_{v\in \+R} \abs{([q]\setminus X) \setminus \omega_{\ol \gamma}(N(v)\cap V(\ol \gamma))}},
\end{align*}
where $V(\ol\gamma)$ is the set of vertices of the subgraph $\ol \gamma$.
The partition function of the polymer model $(\Gamma^*, w)$ on the graph $G^2$ is the following sum:
\begin{align*}
  \Xi(z) = \sum_{\Gamma\in\+S(\Gamma^*)}\prod_{\gamma\in\Gamma}w(\gamma,z).
\end{align*}
Recall that two polymers $\gamma_1$ and $\gamma_2$ are compatible if $d_{G^2}(\ol{\gamma_1},\ol{\gamma_2}) > 1$ and this condition is equivalent to $d_G(\ol{\gamma_1},\ol{\gamma_2})> 2$.
We also extend the definition of $\+C_\gamma(G)$ to $\Gamma\in \+S(\Gamma^*(G))$:
\begin{align*}
  \clusterOfColoring[\Gamma] = \set{\sigma\in \+C_X(G)\cmid
  \sigma|_{\ol{\Gamma}} = \conf[\Gamma]
  \land \sigma(\+L \setminus \ol\Gamma) \subseteq X
  \land \sigma(\+R \setminus \ol\Gamma) \subseteq [q]\setminus X}.
\end{align*}

\begin{lemma} \label{lem:exactRepOfColorings} 
  For $q\ge 3$, all bipartite graphs $G=(\+L,\+R,E)$ with $n$ vertices on both sides and $\emptyset\subsetneq X\subsetneq [q]$,
  \begin{align}
  \abs{\clusterOfColoring} = \abs{X}^n(q-\abs{X})^n\sum_{\Gamma\in\+S(\Gamma^*_X(G)): \aabs{\Gamma} < \alpha n} \prod_{\gamma\in \Gamma} w(\gamma, 1).\label{eq:exact-rep}
  \end{align}
\end{lemma}

\begin{proof}
  Rewrite the right hand side of \Cref{eq:exact-rep} as
  \begin{align*}
    \-{RHS} =  \sum_{\Gamma\in\+S(\Gamma^*)\cmid\aabs{\Gamma} < \alpha n}\abs{X}^n(q-\abs{X})^n \prod_{\gamma\in \Gamma} w(\gamma, 1)
    =\sum_{\Gamma\in\+S(\Gamma^*)\cmid\aabs{\Gamma} < \alpha n} \abs{\+C_\Gamma},
  \end{align*}
  where the last step follows from \Cref{lem:polymerIntuition}. It is now sufficient to show that the set
  \begin{align*}
    \+P\triangleq \set{\+C_\Gamma\cmid \Gamma\in\+S(\Gamma^*)\land \aabs{\Gamma} < \alpha n}
  \end{align*}
  is a partition of $\+C_X$.
  It follows from the definition of $\+C_\Gamma$ that $\+C_{\Gamma_1} \cap \+C_{\Gamma_2}= \emptyset$ if $\Gamma_1\neq \Gamma_2$.
  For any $\sigma\in \+C_X$, it follows from the definition of $p(\sigma)$ that $p(\sigma)$ is compatible and $\aabs{p(\sigma)} < \alpha n$, which shows that $p(\sigma)\in \+P$ and thus $\+C_X\subseteq \cup_{\+C_\Gamma\in \+P}\+C_\Gamma$.
  For any $\sigma\in \+C_\Gamma \in \+P$, it follows from the definition of $\+C_\Gamma$ that $d_X(\sigma)< \alpha n$, which implies that $\sigma\in \+C_X$ and thus $\cup_{\+C_\Gamma\in \+P}C_\Gamma\subseteq \+C_X$.
\end{proof}


\begin{lemma} \label{lem:approxRepOfColorings}
  For $q\ge 3$ and $\Delta\ge\lowerBoundOfDeltaForColorings$, there are constants $C=C(q) > 1$ and $N=N(q)$ such that for all $G\in \graphClassOfColorings$ with $n>N$ vertices on both sides and $X\subseteq [q]$ with $\abs{X}\in \set{\ul q, \ol q}$,
  \begin{align*}
    \abs{X}^n(q-\abs X)^n \Xi(1)=
    \abs{X}^n(q-\abs{X})^n\sum_{\Gamma\in\+S(\Gamma^*_X(G))} \prod_{\gamma\in \Gamma} w(\gamma, 1)
  \end{align*}
  is a $C^{-n}$-relative approximation to $\abs{\clusterOfColoring}$.
\end{lemma}

\begin{proof}
  Clearly $\abs{\+C_X} \ge \ul q^n \ol q^n$.
  Then using \Cref{lem:exactRepOfColorings} and \Cref{lem:decayRateOfColorings} we obtain
  \begin{align}
    \label{eq:midstep-3}
    \frac{\abs{X}^n(q-\abs X)^n \Xi(1)-\abs{\+C_X}}{\abs{\+C_X}}
    &\le \sum_{\Gamma\in\+S(\Gamma^*)\cmid \aabs\Gamma \ge \alpha n}\prod_{\gamma\in \Gamma}w(\gamma,1)\\
    &\le  \sum_{\Gamma\in\+S(\Gamma^*)\cmid \aabs\Gamma \ge \alpha n}(1-1/\ol q)^{(\beta-1)\aabs\Gamma}. \nonumber
  \end{align}
  To enumerate each $\Gamma\in\+S(\Gamma^*)$ with $\aabs{\Gamma}\ge \alpha n$ at least once, we first enumerate an integer $\alpha n \le k \le 2n$, then we choose $k$ first vertices from $\+L\cup \+R$ and enumerate all possible colorings over these $k$ vertices.
  Therefore
  \begin{align*}
    \text{\Cref{eq:midstep-3}}
    \le\sum_{k=\ceil{\alpha n}}^{2n} \binom{2n}{k} \ol q^k(1-1/\ol q)^{(\beta-1)k}
    &\le \sum_{k=\ceil{\alpha n}}^{2n} 2^{H(k/(2n))2n} \ol q^k(1-1/\ol q)^{(\beta-1)k} \\
    &\le \sum_{k=\ceil{\alpha n}}^{2n} \tp{4^{\sqrt{2n/k}} \ol q (1-1/\ol q)^{\beta - 1}}^k \\
    &\le \sum_{k=\ceil{\alpha n}}^{2n} \tp{4^{\sqrt{2/\alpha}} \ol q (1- 1/ \ol q)^{\beta-1}}^k,
  \end{align*}
  where the inequalities follow from \Cref{lem:boundsOfBinomCoefficients} and \Cref{lem:polyUpperBoundOfEntropy}.
  Recall that $\alpha = \valueOfAlphaForColorings$ and $\beta = \valueOfBetaForColorings$.
  Let $f(\Delta) = 4^{\sqrt{2/\alpha}} \ol q (1-1/\ol q)^{\beta - 1}$.
  Using $\Delta\ge\lowerBoundOfDeltaForColorings$, $\ol q \ge 2$, and the inequality $\ln(1+x)\le x$ for any $x>-1$,  we obtain
  \begin{align*}
    f(\Delta)
    &\le
    \exp\tp{\sqrt{2}\Delta^{1/4}\ln 4 + \ln \ol q - \tp{\frac{\Delta^{1/2}}{3}- 1}\frac{1}{\ol q}} \\
    &= \exp\tp{\Delta^{1/4}\tp{\sqrt{2}\ln 4 - \frac{\Delta^{1/4}}{3\ol q}} + \ln \ol q + \frac{1}{\ol q}} \\
    &\le \exp\tp{\Delta^{1/4}\tp{\sqrt{2}\ln 4 - \frac{\sqrt{10}}{3}\ol q \sqrt{\ol q}} + \ln \ol q + \frac{1}{\ol q}} \\
    &\le \exp\tp{\Delta^{1/4}\tp{\sqrt{2}\ln 4 - \frac{2}{3}\sqrt{20}} + \ln \ol q + \frac{1}{\ol q}}.
  \end{align*}
  Since $\sqrt{2}\ln 4 - \frac{2}{3}\sqrt{20} \approx -1.02 < -1$, we obtain
  \begin{align*}
    f(\Delta)
    \le \exp\tp{-\Delta^{1/4} + \ln \ol q + 1/\ol q}
    &\le \exp\tp{-\sqrt{10}\ol q^2\sqrt{\ol q} + \ln \ol q + 1/ \ol q} \\
    &\le \exp\tp{-\sqrt{10}\times 4 \times \sqrt{2} + \ln 2 + 1/2} \\
    &\approx \exp\tp{-16.7} <1.
  \end{align*}
  Therefore, we have
  \begin{align*}
    \text{\Cref{eq:midstep-3}}
    \le \sum_{k=\ceil{\alpha n}}^\infty f(\Delta) ^k
    \le \frac{f(\Delta)^{\alpha n}}{1-f(\Delta)}
    \le \tp{\frac{f(\Delta)^{-\alpha}}{\tp{1-f(\Delta)}^{1/n}}}^{-n}
    < C^{-n}
  \end{align*}
  for some constant $C=C(q)>1$ and for all $n > N$ where $N=N(q)$ is a sufficiently large constant.
  Using the upper bound on \Cref{eq:midstep-3} and $1+x\le \exp(x)$ for any $x\in \=R$ we obtain
  \begin{align*}
    \abs{\+C_X} \le \abs{X}^n(q-\abs X)^n \Xi(1)
    = \abs{\+C_X} + \tp{\abs{X}^n(q-\abs X)^n \Xi(1) -\abs{\+C_X}}
    \le \exp(C^{-n})\abs{\+C_X}
  \end{align*}
  for all $n > N$.
\end{proof}



\begin{lemma}\label{lem:polymerIntuition}
  For $q\ge 3$, all bipartite graphs $G=(\+L,\+R,E)$ with $n$ vertices on both sides, $\emptyset\subsetneq X\subsetneq [q]$ and $\Gamma\in\+S(\Gamma^*_X(G))$,
  \begin{align}
    \abs{X}^n(q-\abs X)^n \prod_{\gamma\in\Gamma} w(\gamma,1) = \abs{\clusterOfColoring[\Gamma]}. \label{eq:midstep-1}
  \end{align}
\end{lemma}

\begin{proof}
  For any $\gamma\in \Gamma$, let $V_\gamma = \ol\gamma\sqcup N_G(\ol\gamma)$. It holds that
  \begin{align}
    w(\gamma,1)=
    \frac{\abs{\clusterOfColoring[\gamma]}}{\abs{X}^n (q-\abs X)^n} =
    \frac{\abs{\+C_{\gamma}(G[V_\gamma])}}{\abs{X}^{\abs{V_\gamma\cap \+L}}(q-\abs{X})^\abs{V_\gamma\cap \+R}},\label{eq:midstep-2}
  \end{align}
  where $\+C_\gamma(G[V_\gamma])$ is the set of colorings $\sigma\in [q]^{V_\gamma}$ that is proper in the graph $G[V_\gamma]$, $\sigma_{\ol \gamma}=\conf$, $\sigma(N(\ol \gamma)\cap \+L)\subseteq X$ and $\sigma(N(\ol \gamma)\cap \+R)\subseteq [q]\setminus X$.
  Since $\Gamma$ is compatible, for any different $\gamma_1\in \Gamma$ and $\gamma_2\in\Gamma$, it holds that $d_G(\ol{\gamma_1},\ol{\gamma_2}) >2$ and thus $V_{\gamma_1}\cap V_{\gamma_2}=\emptyset$.
  Let $l = n - \abs{(\sqcup_{\gamma\in\Gamma}V_\gamma)\cap \+L}$ and $r= n - \abs{(\sqcup_{\gamma\in\Gamma}V_\gamma)\cap \+R}$. Then we have
  \begin{align*}
    \abs{\clusterOfColoring[\Gamma]}
    &=\abs{X}^l (q-\abs{X})^r \prod_{\gamma\in \Gamma} \abs{\+C_\gamma(G[V_\gamma])} \\
    &=\abs{X}^n (q-\abs{X})^n \prod_{\gamma\in\Gamma}\frac{\abs{\+C_{\gamma}(G[ V_\gamma])}}{\abs{X}^{\abs{V_\gamma\cap \+L}}(q-\abs{X})^\abs{V_\gamma\cap \+R}}\\
    &=\abs{X}^n (q-\abs{X})^n \prod_{\gamma\in \Gamma}w(\gamma,1),
  \end{align*}
  where the first step follows from the definition of $\+C_\gamma(G[V_\gamma])$, the second step follows from that $V_{\gamma_1}\cap V_{\gamma_2} = \emptyset$ for any different $\gamma_1,\gamma_2\in \Gamma$ and the last step follows from \Cref{eq:midstep-2}.
\end{proof}

\begin{lemma} \label{lem:decayRateOfColorings}
  For $q\ge 3, \Delta\ge \lowerBoundOfDeltaForColorings, G\in \graphClassOfColorings, \emptyset\subsetneq X\subsetneq [q]$ with $\abs{X}\in \set{\ul q, \ol q}$ and $\gamma\in \Gamma^*(G)$,
  \begin{align*}
    w(\gamma,1) \le
    \tp{1-1/{\ol q}}^{(\beta-1)\aabs{\gamma}}.
  \end{align*}
  As a corollary, for any compatible $\Gamma\subseteq \Gamma^*(G)$,
  \begin{align*}
    \prod_{\gamma\in \Gamma}w(\gamma, 1) \le
    \tp{1-1/{\ol q}}^{(\beta-1)\aabs{\Gamma}}.
  \end{align*}
\end{lemma}

\begin{proof}
  With out loss of generality, we fix $\emptyset\subsetneq X\subsetneq [q]$ with $\abs X = \ul q$ and the other case (if exist) is symmetric.
  Fix $\gamma\in\Gamma^*$.
  Since $G$ is an $(\alpha,\beta)$-expander and $\abs{\ol\gamma} \le \alpha n$, it follows from \Cref{lem:betaMinus1Expansion} that $\abs{N(\ol\gamma)} \ge (\beta-1)\abs{\ol \gamma}$.
  Let $l = \abs{N(\ol\gamma)\cap \+L}$ and $r = \abs{N(\ol\gamma)\cap \+R}$. Then
  \begin{align*}
    w(\gamma,1)
    = \frac{\abs{\clusterOfColoring[\gamma]}}{\abs{X}^n(q-\abs X)^n}
    \le \frac{\ul q^{n-l}(\ul q - 1)^l \ol q^{n-r} (\ol q - 1)^r}{\ul q^n \ol q^n}
    &\le (1-1/\ol q)^{l+r}\\
    &\le (1-1/\ol q)^{(\beta-1) \aabs{\gamma}}.
  \end{align*}
  For any compatible $\Gamma$, it holds that $\aabs{\Gamma} = \sum_{\gamma\in \Gamma}\aabs{\gamma}$. Thus
  \[
    \prod_{\gamma\in \Gamma} w(\gamma,1)\le \prod_{\gamma\in\Gamma}(1-1/\ol q)^{(\beta -1)\aabs{\gamma}}= (1-1/\ol q)^{(\beta-1)\aabs{\Gamma}}.
    \qedhere
  \]
\end{proof}

\begin{lemma}\label{lem:betaMinus1Expansion}
  For $\Delta \ge 3$ and $G=(\+L,\+R,E)\in \graphClassAtHighFugacity$ with $n$ vertices on both sides, $\abs{N_G(U)}\ge (\beta - 1)\abs{U}$ for all $U\subseteq \+L\cup\+R$ with $\abs{U}\le \alpha n$.
\end{lemma}

\begin{proof}
  It follows from the expansion property that
  \begin{align*}
    \abs{N(U)} &= \abs{N(U\cap \+L)\setminus U} + \abs{N(U\cap \+R)\setminus U} \\
    &\ge \tp{\abs{N(U\cap \+L)} - \abs{U\cap \+R}} + \tp{\abs{N(U\cap \+R)} - \abs{U\cap \+L}} \\
    &\ge \tp{\beta \abs{U\cap \+L} - \abs{U\cap \+R}} + \tp{\beta \abs{U\cap \+R} - \abs{U\cap \+L}} \\
    &= (\beta - 1)\abs{U}. \qedhere
  \end{align*}
\end{proof}

\subsection{Approximating the partition function of the polymer model}


\begin{lemma}\label{lem:algOfColoringsPolymerPF}
  For $q\ge 3$ and $\Delta \ge \lowerBoundOfDeltaForColorings$, there is an FPTAS for $\Xi(1)$ for all $G\in \graphClassOfColorings$ and $X \subseteq [q]$ with $\abs{X}\in \set{\ul q, \ol q}$.
\end{lemma}

\begin{proof}
  We use the FPTAS in \Cref{thm:alg} to design the FPTAS we need.
  To this end, we generate a graph $G^2$ in polynomial time in $\abs{G}$ for any $G\in\graphClassOfColorings$.
  We use this new graph $G^2$ as input to the FPTAS in \Cref{thm:alg}.
  It is straightforward to verify the first three conditions in \Cref{thm:alg}, only with the exception that the information of $G^2$ may not be enough because certain connectivity information in $G$ is discarded in $G^2$.
  Nevertheless, we can use the original graph $G$ whenever needed and thus the first three conditions are satisfied.
  For the last condition, \Cref{lem:zeroFreeOfColoringsPolymerPF} verifies it.
\end{proof}

\begin{lemma} \label{lem:zeroFreeOfColoringsPolymerPF}
  There is a constant $R>1$ such that for all $q\ge 3$, $\Delta\ge \lowerBoundOfDeltaForColorings$, $G\in\graphClassOfColorings$ and $X\subseteq [q]$ with $\abs{X}\in \set{\ul q, \ol q}$, $\Xi(z)\neq 0$ for all $z\in \=C$ with $\abs{z} < R$.
\end{lemma}

\begin{proof}
  Set $R=2$.
  For any $\gamma\in\Gamma^*$, let $a(\gamma)=\aabs{\gamma}$.
  We will verify that the KP-condition
  \begin{align}
    \sum_{\gamma\cmid \gamma\not\sim\gamma^*} e^{\aabs\gamma} \abs{w(\gamma,z)} \le \aabs{\gamma^*} \label{eq:KP-condition-1}
  \end{align}
  holds for any $\gamma^*\in \Gamma^*$ and any $\abs{z} < R$.
  It then follows from \Cref{lem:KP-condition} that $\Xi(z)\neq 0$ for any $\abs{z} < R$.
  Fix $\gamma^*\in \Gamma^*$.
  Recall that $d_{G^2}(\ol{\gamma}, \ol{\gamma^*}) \le 1$ for all $\gamma\not\sim \gamma^*$.
  Thus there is always a vertex $v\in \ol \gamma$ such that $v\in \ol{\gamma^*}\sqcup N_{G^2}(\ol{\gamma^*})$.
  The number of such vertices $v$ is at most $(\Delta^2+1) \aabs{\gamma^*}$.
  So to enumerate each $\gamma\neq \gamma^*$ at least once, we can
  \begin{enumerate}[label=\alph*)]
    \item first enumerate a vertex $v\in \ol{\gamma^*}\sqcup N_{G^2}(\ol{\gamma^*})$;
    \item then enumerate an integer $k$ from $1$ to $\floor{\alpha n}$;
    \item finally enumerate $\gamma$ with $v\in \ol \gamma$ and $\ol{\gamma}=k$.
  \end{enumerate}
  Since $\ol{\gamma}$ is connected in $G^2$, applying \Cref{lem:EnumSubgraph} and using \Cref{lem:decayRateOfColorings} to bound $\abs{w(\gamma,z)}$ we obtain
  \begin{align}
    \label{eq:midstep13}
    \sum_{\gamma:\gamma\not\sim\gamma^*} e^{\aabs{\gamma}}\abs{w(\gamma,z)}
    &\le
    (\Delta^2+1)\aabs{\gamma^*}\sum_{k=1}^{\floor{\alpha n}}(e\Delta^2)^{k-1} \ol q^k e^k (1-1/\ol q)^{(\beta - 1)k}\abs{z}^k.
  \end{align}
  Adding some extra nonnegative terms and using $\abs{z} < R$, we obtain
  \begin{align*}
    \text{\Cref{eq:midstep13}}
    &\le \frac{\Delta^2+1}{e\Delta^2}\aabs{\gamma^*}\sum_{k=1}^\infty \tp{e^2\Delta^2 \ol q (1-1/\ol q)^{\beta - 1}R}^k.
  \end{align*}
  Recall that $\beta=\valueOfBetaForColorings$ and $\Delta\ge \lowerBoundOfDeltaForColorings$.
  It holds that
  \begin{align*}
    e^2\Delta^2 \ol q (1-1/\ol q)^{\beta - 1}R
    &=
    \exp\tp{2 + 2\ln \Delta + \ln \ol q + \frac{1}{\ol q} +\ln R - \frac{\Delta^{1/2}}{3\ol q}} \\
    &\le
    \exp\tp{2 + 2\ln 100 + 21\ln \ol q + \frac{1}{\ol q} +\ln R - \frac{10}{3}\ol q^4} \\
    &\le
    \exp\tp{2 + 2\ln 100 + 21\ln 2 + \frac{1}{2} +\ln 1.1 - \frac{10}{3}2^4} \\
    &< 2^{-10},
  \end{align*}
  where the inequalities follow from the monotonicity of corresponding functions.
  Therefore
  \begin{align*}
    \text{\Cref{eq:midstep13}}
    \le \frac{\Delta^2+1}{e\Delta^2}\aabs{\gamma^*} \sum_{k=1}^\infty 2^{-10k}
    \le 2\aabs{\gamma^*} \frac{2^{-10}}{1-2^{-10}} < \aabs{\gamma^*},
  \end{align*}
  which proves \Cref{eq:KP-condition-1}.
\end{proof}

\subsection{Putting things together}

Using the results from previous parts, we obtain our main result for counting colorings.

{\renewcommand{\thetheorem}{\ref{thm:q-coloring}}
\begin{theorem}
  For $q\ge 3$ and $\Delta \ge \lowerBoundOfDeltaForColorings$, with high probability (tending to $1$ as $n\to \infty $) for a graph chosen uniformly at random from $\setG$, there is an FPTAS to count the number of $q$-colorings.
\end{theorem}
\addtocounter{theorem}{-1} }

\begin{proof}
  This theorem follows from \Cref{lem:almostOfColorings} and \Cref{lem:algOfColorings}.
\end{proof}

\begin{algorithm}[htbp]
  \caption{Counting colorings for $q\ge 3$ and $\Delta\ge \lowerBoundOfDeltaForColorings$}
  \label{alg:Colorings}
  \begin{algorithmic}[1]
    \State \textbf{Input:} \emph{A graph $G=(\+L,\+R,E)\in \graphClassOfColorings$ with $n$ vertices on both sides and $\eps>0$}
    \State \textbf{Output:} \emph{$\widehat Z$ such that $\exp(-\eps)\widehat Z \le \abs{\+C(G)}\le \exp(\eps)\widehat Z$}
    \If {$n\le N$ or $\eps \le 2C^{-n}$}
      \State Use the brute-force algorithm to compute $\widehat Z \gets \abs{\+C(G)}$;
      \State Exit;
    \EndIf
    \State $\eps'\gets \eps - C^{-n}$;
    \State Use the FPTAS in \Cref{lem:algOfColoringsPolymerPF} to obtain $\widehat Z_1$, an $\eps'$-relative approximation to the partition function $\Xi(z)$ at $z=1$ of the polymer model $(\Gamma^*_{[\ul q]}(G),w)$.
    \If {$q$ is even}
      \State $\widehat  Z \gets \binom{q}{\ul q} {\ul q}^{2n}\widehat Z_1$;
    \Else
      \State Use the FPTAS in \Cref{lem:algOfColoringsPolymerPF} to obtain $\widehat Z_2$, an $\eps'$-relative approximation to the partition function $\Xi(z)$ at $z=1$ of the polymer model $(\Gamma^*_{[\ol q]}(G),w)$.
      \State $\widehat Z \gets \binom{q}{\ul q} \tp{\ul q \ol q}^{n}\tp{\widehat Z_1 + \widehat Z_2}$;
    \EndIf
  \end{algorithmic}
\end{algorithm}

\begin{lemma}\label{lem:algOfColorings}
  For $q\ge 3$ and $\Delta\ge \lowerBoundOfDeltaForColorings$, there is an FPTAS for $\abs{\+C(G)}$ for all $G\in \graphClassOfColorings$.
\end{lemma}

\begin{proof}
  First we state our algorithm.
  See \Cref{alg:Colorings} for a pseudocode description.
  Fix $q\ge 3$ and $\Delta \ge \lowerBoundOfDeltaForColorings$.
  The input is a graph $G=(\+L,\+R,E)\in \graphClassOfColorings$ and an approximation parameter $\eps > 0$.
  The output is a number $\widehat Z$ to approximate $\abs{\+C(G)}$.
  We use $\Xi_{1}(z)$ and $\Xi_{2}(z)$ to denote the partition functions of the polymer models $(\Gamma^*_{[\ul q]}(G),w)$ and $(\Gamma^*_{[\ol q]}(G),w)$, respectively.
  Let $N_1,C_2,N_2,C_2$ be the constants in \Cref{lem:strucApproxOfColorings} and \Cref{lem:approxRepOfColorings}, respectively.
  Let $Z= \binom{q}{\ul q}\abs{\+C_{[\ul q]}(G)}$ if $q$ is even, otherwise $Z = \binom{q}{\ul q}\abs{\+C_{[\ul q]}(G) + \+C_{[\ol q]}(G)}$.
  These two lemmas show that $Z$ is a $C_1^{-n}+C_2^{-n}\le 2\min(C_1,C_2)^{-n}\le C^{-n}$-relative approximation to $\abs{\+C(G)}$ for another constant $C > 1$ and all $n > N \ge \max(N_1,N_2)$ where $N$ is another sufficiently large constant.
  If $n\le N$ or $\eps \le 2C^{-n}$, we use the brute-force algorithm to compute $\abs{\+C(G)}$.
  If $\eps > 2C^{-n}$, we apply the FPTAS in \Cref{lem:algOfColoringsPolymerPF} with approximation parameter $\eps' = \eps-C^{-n}$ to obtain $\widehat Z_1$, an $\eps'$-relative approximation to $\Xi_1(1)$. If $q$ is even, then $\widehat Z=\binom{q}{\ul q} {\ul q}^{2n}\widehat Z_1$ is the output of the algorithm. Otherwise, we apply again the FPTAS in \Cref{lem:algOfColoringsPolymerPF} with approximation parameter $\eps' = \eps-C^{-n}$ to obtain $\widehat Z_2$, an $\eps'$-relative approximation to $\Xi_2(1)$. And the output is $\widehat Z = \binom{q}{\ul q} \tp{\ul q \ol q}^{n}\tp{\widehat Z_1 + \widehat Z_2}$. It is clear that $\exp(-\eps) \widehat Z \le \abs{\+C(G)}\le \exp(\eps)\widehat Z$.

  Then we show that \Cref{alg:Colorings} is indeed an FPTAS.
  It is required that the running time of our algorithm is bounded by $\tp{n/\eps}^{C_3}$ for some constant $C_3$ and for all $n > N_3$ where $N_3$ is a constant.
  Let $N_3=N$.
  If $\eps \le 2C^{-n}$, the running time of the algorithm would be $q^n \le (nC^n/q)^{C_3} \le \tp{n/\eps}^{C_3}$ for sufficient large $C_3$.
  If $\eps > 2C^{-n}$, the running time of the algorithm would be $\tp{n/\eps'}^{C_4}=\tp{n/(\eps - C^{-n})}^{C_4}\le \tp{2n/\eps}^{C_4} \le \tp{n/\eps}^{C_3}$ for sufficient large $C_3$, where $C_4$ is a constant from the FPTAS in \Cref{lem:algOfColoringsPolymerPF}.
\end{proof}


\bibliographystyle{alpha}
\bibliography{refs}

\end{document}